%% file: main.tex
\documentclass[sigconf]{acmart}

\usepackage{amssymb,amsfonts,amsmath,amsthm,amscd,dsfont,mathrsfs,mathtools}

\input{commands}

\usepackage{graphicx}
\usepackage{appendix}
\usepackage{float}
\usepackage{xcolor}
\usepackage{amsthm}
\usepackage{geometry,enumitem}

\renewcommand\footnotetextcopyrightpermission[1]{}  
\begin{document}
\fancyhead{}


\title{Everything is a Race and Nakamoto Always Wins}
%

\author{Amir Dembo}
\affiliation{Stanford University}
\email{amir@math.stanford.edu}

\author{Sreeram Kannan}
\affiliation{University of Washington}
\email{ksreeram@uw.edu}

\author{Ertem Nusret Tas}
\affiliation{Stanford University}
\email{nusret@stanford.edu}

\author{David Tse}
\affiliation{Stanford University}
\email{dntse@stanford.edu}

\author{Pramod Viswanath}
\affiliation{University of Illinois Urbana-Champaign}
\email{pramodv@illinois.edu}

\author{Xuechao Wang}
\affiliation{University of Illinois Urbana-Champaign}
\email{xuechao2@illinois.edu}
 
\author{Ofer Zeitouni}
\affiliation{Weizmann Institute of Science}
\email{ofer.zeitouni@weizmann.ac.il}

\thanks{The authors are listed alphabetically. For correspondence on the paper, please contact DT at dntse@stanford.edu.}





\begin{abstract}
Nakamoto invented the longest chain protocol, and claimed its security by analyzing the private double-spend attack, a race between the adversary and the honest nodes to grow a longer chain. But is it the worst attack? We answer the question in the affirmative for three classes of longest chain protocols, designed for different consensus models: 1) Nakamoto's original Proof-of-Work protocol; 2) Ouroboros and SnowWhite Proof-of-Stake protocols; 3) Chia Proof-of-Space protocol. As a consequence, exact characterization of the maximum tolerable adversary power is obtained for each protocol as a function of the average block time normalized by the network delay. The security analysis of these protocols is performed in a unified manner by a novel method of reducing all attacks to a race between the adversary and the honest nodes.


\end{abstract}

\maketitle
%
%
%

\input{introduction}
\input{model}

\input{nakamoto_blocks}
\input{analysis}







\input{sample_path}

\section{Acknowledgments}

Amir Dembo and Ofer Zeitouni were partially supported  by a US-Israel BSF grant.
Ertem Nusret Tas was supported in part by the Stanford Center for Blockchain Research.
This research is also supported in part by NSF under grants CCF-1705007, DMS-1954337, 1651236 and
Army Research Office under grant W911NF-14-1-0220.
We thank the reviewers for the helpful comments.


\bibliographystyle{alpha}
\bibliography{references}

\appendices
\section*{Appendix}

\input{proofs}
\input{Appendix_Nakamoto}
\input{Appendix_PoW_Ouroboros}

\input{Appendix_chia}

\input{Appendix_per_live}
\input{Appendix_Worst_Attack}

\end{document}

%% file: commands.tex
\usepackage{comment}
\usepackage{graphicx,subcaption}
\usepackage{varwidth}
\usepackage{IEEEtrantools}
\usepackage[shortlabels]{enumitem}

\newtheorem{thm}{Theorem}[section]
\newtheorem{defn}[thm]{Definition}

\newcommand{\bpf}{\begin{proof}}
\newcommand{\epf}{\end{proof}}

\newcommand{\beqa}{\begin{eqnarray}}

\newcommand{\eeqa}{\end{eqnarray}}

\newcommand{\beqan}{\begin{eqnarray*}}

\newcommand{\eeqan}{\end{eqnarray*}}

\newcommand{\beq}{\begin{equation}}

\newcommand{\eeq}{\end{equation}}

\renewcommand{\C}{{\cal C}}

\renewcommand{\tt}{\mathcal{T}}

\newcommand{\singlespacing}{\let\CS=\@currsize\renewcommand{\baselinestretch}{0.95}\tiny\CS}

\newcommand{\oneandahalfspacing}{\let\CS=\@currsize\renewcommand{\baselinestretch}{1.25}\tiny\CS}

\newcommand{\doublespacing}{\let\CS=\@currsize\renewcommand{\baselinestretch}{1.39}\tiny\CS}

\newcommand{\be}{\begin{equation}}

\newcommand{\ee}{\end{equation}}

\newcommand{\bc}{\begin{center}}

\newcommand{\ec}{\end{center}}

\newcommand{\bfl}{\begin{flushleft}}

\newcommand{\efl}{\end{flushleft}}

\makeatletter
\renewcommand\subsubsection{\@startsection{subsubsection}{3}{\z@}%
                       {-18\p@ \@plus -4\p@ \@minus -4\p@}%
                       {4\p@ \@plus 2\p@ \@minus 2\p@}%
                       {\normalfont\normalsize\bfseries\boldmath
                        \rightskip=\z@ \@plus 8em\pretolerance=10000 }}

\renewcommand{\H}{{\mathcal{H}}}
\renewcommand{\C}{{\mathcal{C}}}

\newcommand{\T}{{\mathcal{T}}}

\makeatother

\renewcommand{\epsilon}{\varepsilon}

\renewcommand{\tt}{\mathcal{T}}
\newcommand{\zz}{\mathcal{Z}}

\newcommand{\I}{{\mathcal I}}
\newcommand{\E}{{\mathcal E}}
\newcommand{\W}{{\mathcal W}}
\newcommand{\J}{{\mathcal J}}

\usepackage{algorithm}
\usepackage[noend]{algpseudocode}
\usepackage{xcolor}
\definecolor{azure}{rgb}{0.54, 0.17, 0.89}

\usepackage{enumitem,kantlipsum}

%% file: introduction.tex
\section{Introduction}

\subsection{Background}

In 2008, Satoshi Nakamoto invented the concept of {\em blockchains} as a technology for maintaining decentralized ledgers \cite{bitcoin}. A core contribution of this work is the {\em longest chain} protocol, a deceptively simple consensus algorithm. Although invented in the context of Bitcoin and its Proof-of-Work (PoW) setting, the longest chain protocol has been adopted in many blockchain projects, as well as extended to other more energy-efficient settings such as Proof-of-Stake (PoS) (eg. \cite{bentov2016snow}, \cite{kiayias2017ouroboros},\cite{david2018ouroboros},\cite{badertscher2018ouroboros},\cite{fan2018scalable}) and Proof-of-Space (PoSpace) (eg. \cite{abusalah2017beyond,cohen2019chia,park2018spacemint}). 

Used to maintain a ledger for a valued asset in a permissionless environment, the most important property of the longest chain protocol is its {\em security}: how much resource does an adversary need to attack the protocol and revert transactions already confirmed? Nakamoto analyzed this property by proposing a specific attack: the private double-spend attack (Figure \ref{fig:tree_partition}(a)). The adversary grows a private chain of blocks in a race to attempt to outpace the public longest chain and thereby replacing it after a block in the public chain becomes  $k$-deep. Let $\lambda_h$ and $\lambda_a$ be the rate at which the honest nodes and the adversary mine blocks, proportional to their respective hashing powers. Then it is clear from a law of large numbers argument that if $\lambda_a > \lambda_h$, then the adversary will succeed with high probability no matter how large $k$ is. Conversely, if $\lambda_a < \lambda_h$, the probability of the adversary succeeding decreases exponentially with $k$.  When there is a network delay of $\Delta$ between honest nodes, this condition for security becomes:
\begin{equation}
\label{eq:private_threhold}
    \lambda_a < \lambda_{{\rm growth}}(\lambda_h, \Delta),
\end{equation}
where $\lambda_{{\rm growth}}(\lambda_h, \Delta)$ is the growth rate of the honest chain under worst-case forking. In a fully decentralized setting with many honest nodes each having small mining power, \cite{ghost} calculates this to be $\lambda_{\rm growth} = \lambda_h/(1+\lambda_h \Delta)$. If we let $\beta$ to be the adversary fraction of power, then (\ref{eq:private_threhold}) yields the following condition:
\begin{equation}
\label{eq:private2}
    \beta < \frac{1-\beta}{1+(1-\beta) \lambda \Delta}.
\end{equation}
Here, $\lambda$ is the total mining rate, and $\lambda \Delta$ is the number of blocks mined per network delay. $1/(\lambda \Delta)$ is the block speed normalized by the network delay. Solving (\ref{eq:private2}) at equality gives a security threshold $\beta_{\rm pa}(\lambda \Delta)$. When $\lambda \Delta$ is small,  $\beta_{pa}(\lambda\Delta) \approx 0.5$, and this leads to Nakamoto's main claim in \cite{bitcoin}: the longest chain protocol is secure as long as the adversary has less than $50\%$ of the total hashing power and the mining rate is set to be low. A more aggressive mining rate to speed up the blockchain reduces the security threshold. Hence (\ref{eq:private2}) can be viewed as a tradeoff between security and block speed.

The private double-spend attack is a {\em specific} attack, and  Nakamoto claimed security based on the analysis of this attack alone. But what about other attacks? Are there other worse attacks? A pertinent question after Nakamoto's work is the identification of the {\em true} security threshold $\beta^*(\lambda \Delta)$ in the face of the {\em worst} attack. The groundbreaking work \cite{backbone} first addressed  this question by formulating and performing a formal security analysis of the Proof-of-work longest chain protocol. They used a lock-step round-by-round synchronous model, and the analysis was later extended to the more realistic $\Delta$-synchronous model \cite{pss16}. The results show that when $\lambda \Delta \rightarrow 0$, indeed $\beta^*(\lambda \Delta)$ approaches $50\%$, thus validating Nakamoto's intuition. However, for $\lambda  \Delta > 0$, there is a gap between their bounds and the private attack security threshold, and this gap grows when $\lambda\Delta$ grows. 


\subsection{Main contribution}

\begin{figure}
    \centering
    \includegraphics[width=0.45\textwidth]{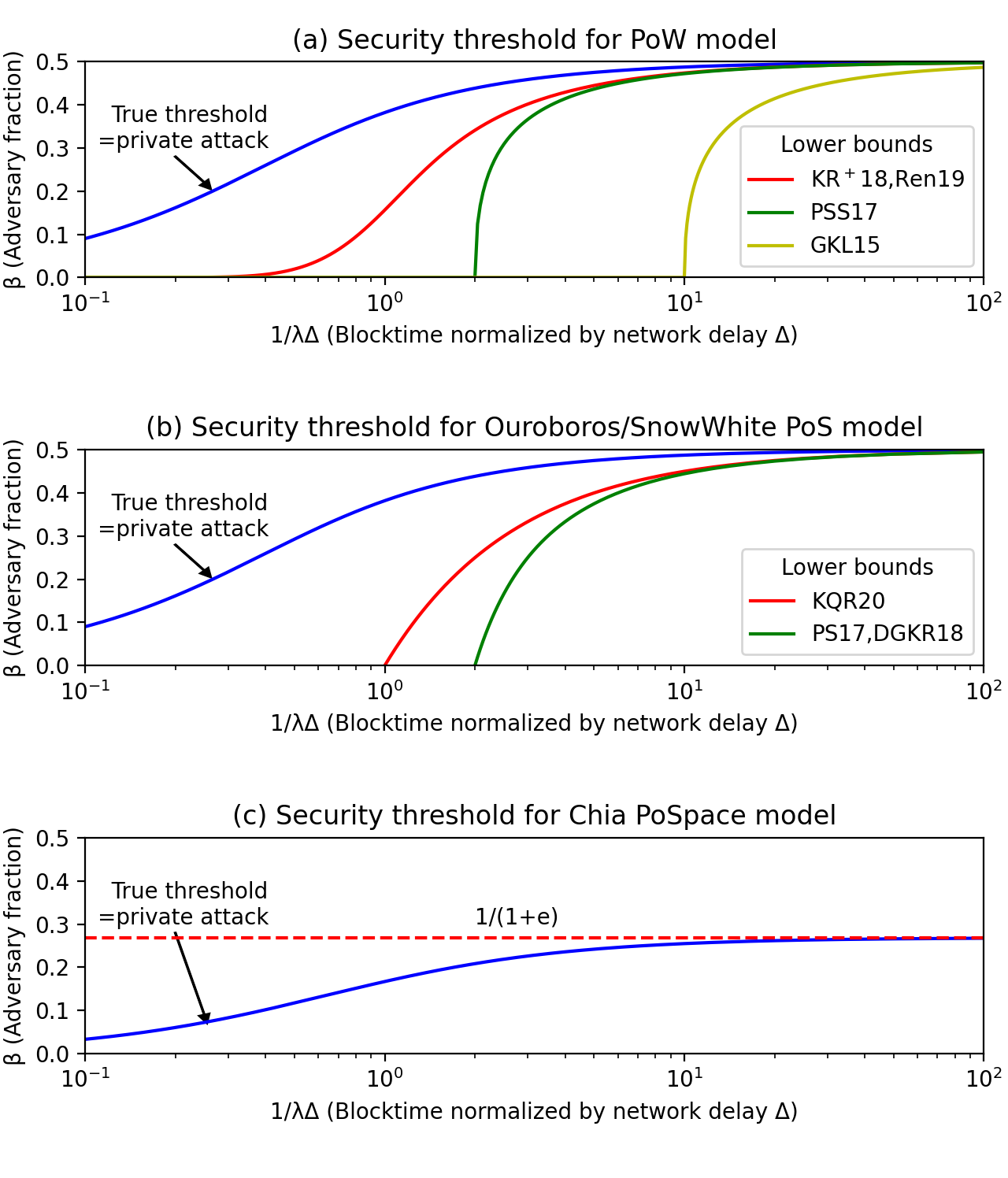}
    \caption{True security threshold as a function of normalized block speed, compared to bounds in the literature. (a) Proof-of-work model; (b)  Ouroboros/SnowWhite Proof-of-Stake model; (c) Chia Proof-of-Space model. In (a) and (b), the blue curve represents $\beta^*(\lambda\Delta) = \beta_{\rm pa}(\lambda\Delta)$; both PoW and PoS have the same (true) security threshold. In (a), the red, green and yellow curves are obtained by solving $\beta = (1-\beta)e^{-2(1-\beta)\lambda\Delta}$, $\beta = (1-\beta)(1-2\lambda\Delta(1-\beta))$ and $\beta = (1-\beta)(1-10\lambda\Delta(1-\beta))$ respectively. In (b), the red and green curves are $(1-\beta)/(1+\lambda\Delta) = 1/2$ and $(1-\beta)(1-\lambda\Delta) = 1/2$ respectively. In (c), the blue curve is the solution of $ e\beta = \frac{1-\beta}{1+(1-\beta)\lambda \Delta}$, the true threshold, and also that of private attack. Unlike in (a) and (b), the true threshold does not reach $0.5$ when $\lambda \Delta \rightarrow 0$, but reach $1/(1+e)$ instead. Note that while in all cases , the true security threshold equals the private attack threshold, the threshold is different for Chia than for the other two.}
    \label{fig:security_threshold}
\end{figure}

The main contribution of this work is a new approach to the security analysis of longest chain protocols. This approach is driven by the question of whether the private attack is the worst attack for longest chain protocols in a broad sense. Applying this approach to analyze three classes of longest chain protocols in the $\Delta-$synchronous model\cite{pss16}, we answer this question in the affirmative in all cases: {\bf the {\em true} security threshold is the same as the private attack threshold}:
\begin{equation}
    \beta^*(\lambda \Delta) = \beta_{\rm pa}(\lambda \Delta) \quad \mbox{for all $\lambda \Delta \ge 0$}
\end{equation} (Figure \ref{fig:security_threshold}). The three classes are: 1) the original Nakamoto PoW protocol; 2) Ouroboros Praos \cite{david2018ouroboros} and SnowWhite \cite{sleepy,bentov2016snow} PoS protocols; 3) Chia PoSpace protocol \cite{cohen2019chia}. They all use the longest chain rule but differ in how the lotteries for proposing blocks are run. (Figure \ref{fig:models})    In the first two protocols, we close the gap between existing bounds and the private attack threshold by identifying the true threshold to be the private attack threshold at all values of $\lambda\Delta$.  For Chia, the adversary is potentially very powerful, since at each time, the adversary can mine on every block of the blocktree, and each block provides an independent opportunity for winning the lottery.  It was not known to be secure for {\em any} non-zero fraction of adversary power. (More specifically, while \cite{cohen2019chia} proved the chain growth and chain quality properties for the Chia protocol, the crucial common prefix property is missing.)  Our result not only says that Chia is secure, but it is secure all the way up to the private attack threshold (although the private attack threshold is smaller for Chia than for the other two classes of protocols due to the increased power of the adversary).

\begin{figure*}
    \centering
    \includegraphics[width=\textwidth]{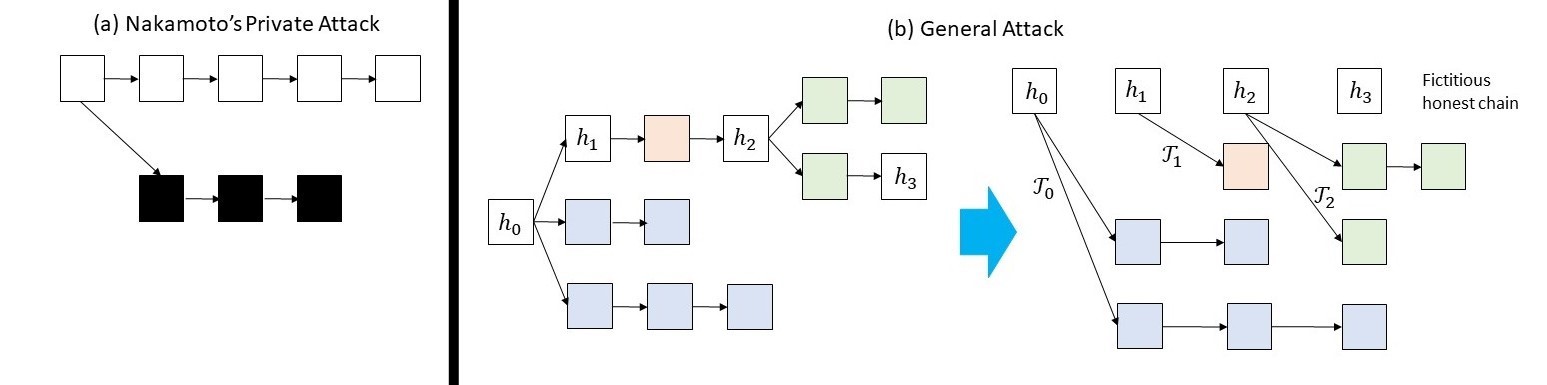}
    \caption{(a)Nakamoto's private attack as a race between a single adversary chain and the honest chain. (b) By blocktree partitioning, a general attack is represented as multiple adversary chains simultaneously racing with a fictitious honest chain. Note that this fictitious chain is formed by only the honest blocks, and may not correspond to the longest chain in the actual system. However, the longest chain in the actual system must grow no slower than this fictitious chain. }
    \label{fig:tree_partition}
\end{figure*}

\begin{figure}
    \centering
    \includegraphics[width=\textwidth/2]{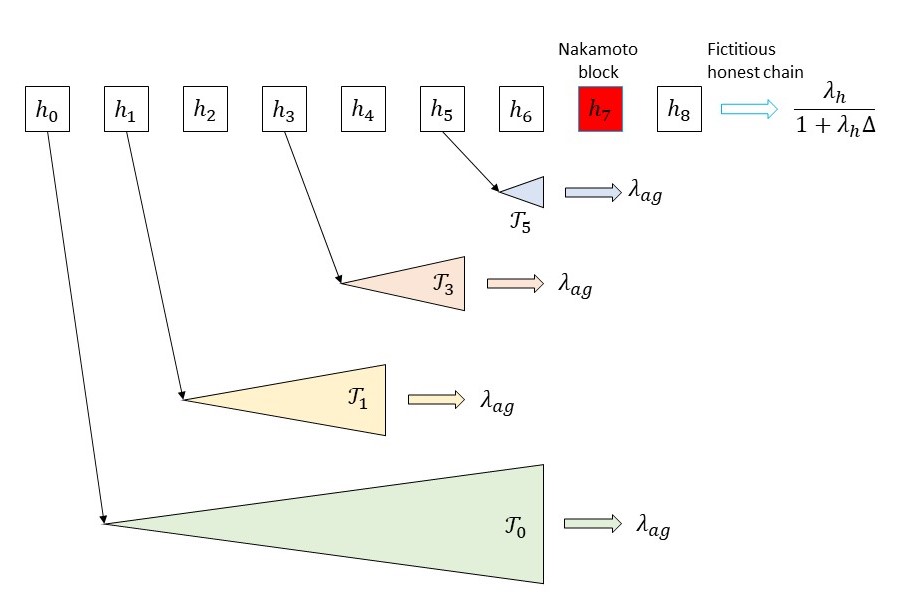}    \caption{Race between the adversary trees and the fictitious honest chain. While there may be multiple adversary trees simultaneously racing with the honest chain, the growth rate of each tree is bounded by the growth rate of the adversary chain in the private attack. An honest block is a Nakamoto block when  all the previous adversary trees  never catch up with the honest chain past that block. } 
    \label{fig:nakamoto_blocks}
\end{figure}

That the true security threshold matches the private attack threshold in all these protocols is not a coincidence. It is due to an intimate connection between the private attack and any general attack. Our approach exposes and exploits this connection by defining two key concepts: {\em blocktree partitioning} and {\em Nakamoto blocks}. Through these concepts, we can view {\em any} attack as a race between adversary and honest chains, not just the private attack. However, unlike the private attack, a general attack may send many adversary chains to simultaneously race with the honest chain.

The entire blocktree, consisting of both honest and adversary blocks, public or private, is particularly simple under a private attack: it can be partitioned into two chains, one honest and one adversary (Figure \ref{fig:tree_partition}(a)). In contrast, under a general attack where the  adversary can make public blocks at multiple time instances, a much more complex blocktree can emerge (Figure \ref{fig:tree_partition}(b)). However, what we observe is that by partitioning this more complex tree into sub-trees, each rooted at a honest block and consisting otherwise entirely of adversary blocks, one can view the general attack as initiating {\em multiple} adversary sub-trees to race with a single fictitious chain consisting of only honest blocks (Figure \ref{fig:nakamoto_blocks}). The growth rate of each of these adversary sub-trees is upper bounded by the growth rate of the adversary chain used in the private attack.  Therefore, if the private attack is unsuccessful, we know that the growth rate of each of the adversary trees must be less than that of the fictitious honest chain. What we show, for each of the three classes of protocols, is that under that condition, there must exist honest blocks, which we call {\em Nakamoto blocks}, each having the property that {\em none} of the past adversary trees can {\em ever} catch up after the honest chain reaches the block. These Nakamoto blocks serve to stabilize the blockchain: when each such block enters the blocktree, complex as it may be, we are guaranteed that the entire prefix of the longest chain  up to that block remains immutable in the future\footnote{Thus, Nakamoto blocks have a god-like permanence, they exist, but nobody knows which block is a Nakamoto block.}. When Nakamoto blocks occur and occur frequently, the persistence and liveness of the blockchain is guaranteed.

\subsection{Related works}

There have been several significant ideas that have emerged from  the security analysis of blockchains in the past few years, and below we  put our contribution in the perspective of these ideas.

\cite{backbone} initiated blockchain security analysis through defining key {\em backbone} properties\footnote{Properties of the blocktree, independent of the content of the blocks.} of chain common prefix, chain quality and chain growth. Applying this framework to analyse the PoW longest chain protocol in the lock-step round-by-round model, it is shown that the common prefix property, the most difficult property to analyze, is satisfied if the number of adversary blocks over a long window is less than the number of uniquely successful honest blocks\footnote{A uniquely successful honest block is one that is the only honest block mined in a round.}. A similar block counting analysis is conducted by \cite{pss16} in the $\Delta-$ synchronous model, with the notion of uniquely successful blocks replaced by the notion of {\em convergence opportunities}. The resulting bound is tight when $\lambda \Delta$ is small but loose in general. Moreover, the block-counting technique completely breaks down for analyzing PoS longest chain protocols because of the notorious {\em Nothing-at-Stake} problem: winning one lottery can yield a very large number of blocks for the adversary. To overcome this issue, two new ideas were invented. In the Ouroboros line of work \cite{kiayias2017ouroboros,david2018ouroboros,badertscher2018ouroboros}, a new notion of {\em forkable strings} was invented and a Markov chain analysis was performed to show convergence of the longest chain regardless of adversary action if the adversary stake is below a certain threshold. Sleepy Consensus and SnowWhite \cite{sleepy,bentov2016snow} took a different approach and defined a notion of {\em a pivot}, which is a time instance $t$ such that in all time intervals around $t$, there are more honest convergence opportunities than the number of adversary slots. They showed that a pivot forces convergence of the longest chain up to that time, and moreover if the adversary stake is less than a certain threshold, then these pivots must occur and they must occur often. 

Despite this impressive stream of ideas, the true security threshold was still unknown for both the PoW and PoS longest chain protocols. Moreover, the analysis techniques seem very tied to the specific longest chain protocol under study.
The definition of a pivot in \cite{sleepy}, for example, is  tied to the specific longest chain protocol, SnowWhite, they designed. 
In contrast, the notion of Nakamoto blocks in our approach can be viewed as a more general notion of pivots, but defined for general longest chain protocols and designed to tie the problem back to the private attack. Even though the analysis method in \cite{sleepy} has already evolved (or, shall we say, pivoted)  from the analysis method in \cite{backbone}, the influence of the block counting method is still felt in the definition of a pivot. We depart from this method by defining a Nakamoto block directly in terms of structural properties of the evolving blocktree itself. In fact, our approach was motivated from analyzing a protocol like Chia, where the rate of adversary winning slots grows exponentially over time and hence a condition like the one used in \cite{sleepy} does not give non-trivial bounds.

The present paper is an extension of an earlier version \cite{pos_paper}, where we introduced and applied this approach to analyze a PoS longest chain protocol \cite{fan2018scalable} similar to the Chia protocol. Since we released that early version, we became aware of an independent work \cite{kiayias2020consistency}, which obtains the true security threshold as well as linear consistency for the Ouroboros Praos protocol in the lock-step round-by-round model. They achieved this by tightening the definition of a pivot in \cite{sleepy} to count all honest slots, including concurrent ones, not only uniquely successful ones. Like the original definition of pivots, however, this definition is tied to the specific protocol. The approach would not give non-trivial bounds for the Chia protocol, for example. Moreover, their result on the Praos protocol under the $\Delta$-synchronous model is not tight (Figure \ref{fig:security_threshold}(b)).  We believe this is due to their analysis technique of mapping the $\Delta$-synchronous model back to the lock-step round-by-round model. In contrast, our analysis is directly in the $\Delta$-synchronous model and yields tight results in that model.

After the initial submission of this paper, we were made aware of independent work \cite{cryptoeprint:2020:661}, which obtained the same results for the PoW and the Ouroboros PoS protocols, but using totally a different set of techniques based on forkable strings. 

\subsection{Outline}

In Section \ref{sec:model}, we introduce a unified model for all three classes of protocols. In Section \ref{sec:nak_blks}, we introduce the central concepts of this work: blocktree partitioning and Nakamoto blocks. These concepts are applicable to any longest chain protocol. In Section \ref{sec:analysis}, we use these concepts in the security analysis of the three classes of protocol  attaining the private attack security threshold of each. In Section \ref{sec:worst_attack} we explore the question of whether the private attack is worst case in a stronger sense for longest chain protocols.

%% file: model.tex
\section{Models}
\label{sec:model}


A key goal of this paper is to provide a common framework to analyze the security properties of various longest chain protocols. We focus here primarily on the graph theoretic and the stochastic aspects of the problem: some resource-dependent randomness is utilized by these protocols to select which node is eligible to create a block. The modality in which the randomness is generated leads to different stochastic processes describing the blocktree growth. Understanding these stochastic processes and the ability of the adversary to manipulate these processes to its advantage is the primary focus of the paper. 

Different longest chain protocols use different cryptographic means to generate the randomness needed. We specifically exclude here the cryptographic aspects of the protocols, whose  analysis is necessary to guarantee the full security of these protocols. In most of the protocols we consider (for example \cite{backbone,kiayias2017ouroboros}), the cryptographic aspects have already been carefully studied in the original papers and are not the primary bottleneck.  In others, further work may be necessary to guarantee the full cryptographic security. In all of these protocols, we assume ideal sources of randomness  to create a model that can then be analyzed independently.  

We will adopt a continuous-time  model, following the tradition set by Nakamoto \cite{bitcoin} and also  used in several subsequent influential works (eg. \cite{ghost}) as well as more recent works (eg. \cite{ren} and \cite{li2020continuoustime}). The continuous-time model affords analytical simplicity and allows us to focus on the essence of the problem without being cluttered by too many parameters.   Our model corresponds roughly to the $\Delta-$synchronous network model introduced in \cite{pss16} in the limit of a large number of lottery rounds over the duration of the network delay. This assumption seems quite reasonable. For example, the total hash rate in today's Bitcoin network is about $100$ ExaHash/s, i.e. solving $10^{21}$ puzzles per second. Nevertheless, we believe our results can be extended to the discrete setting.


We first explain the model in the specific context of Nakamoto's Proof-of-Work longest chain protocol, and then generalize it to a unified model for all three classes of protocols we study in this paper.


\subsection{Modeling proof-of-work longest chain}
\label{sec:model_nak}

The blockchain is run on a network of $n$ honest nodes and a set of malicious nodes. Each honest node mines blocks, adds them to the tip of the longest chain in its current view of the blocktree and broadcasts the blocks to other nodes. Malicious nodes also mine blocks, but they can be mined elsewhere on the blocktree, and they can also be made public at arbitrary times. Due to the memoryless nature of the puzzle solving and the fact that many attempts are tried per second, we model the block mining processes as Poisson with rates proportional to the hashing power of the miner.

Because of network delay, different nodes may have different views of the blockchain. Like the $\Delta$-synchronous model in \cite{pss16}, we assume there is a bounded communication delay $\Delta$ seconds between the $n$ honest nodes. We assume malicious nodes have zero communication delay among each other, and they can always act in collusion, which in aggregate is referred as {\it the adversary}. Also the adversary can delay the delivery of all broadcast blocks by up to $\Delta$ time. Hence, the adversary has the ability to have one message delivered to honest nodes at different times, all of which has to be within $\Delta$ time of each other.

More formally, the evolution of the blockchain can be modeled as a process $\{(\T(t), \C(t), \T^{(p)}(t), \C^{(p)}(t)): t \ge 0, 1 \leq p \leq n\}$, $n$ being the number of honest miners, where:
\begin{itemize}
    \item $\T(t)$ is a tree, and is interpreted as the {\em mother tree} consisting of all the blocks that are mined by both the adversary and the honest nodes up until time $t$, including blocks that are kept in private by the adversary and including blocks that are mined by the honest nodes but not yet heard by other honest nodes in the network. 
    \item $\T^{(p)}(t)$ is an induced (public) sub-tree of the mother tree $\T(t)$ in the view of the $p$-th honest node at time $t$. It is the collection of all the blocks that are mined by node $p$ or received from other nodes up to time $t$.
    \item $\C^{(p)}(t)$ is a longest chain in the tree $\T^{(p)}(t)$, and is interpreted as the longest chain in the local view of the $p$-th honest node on which it is mining at time $t$. Let $L^{(p)}(t)$ denote the depth, i.e the number of blocks in $\C^{(p)}(t)$ at time $t$.
   \item $\C(t)$ is the common prefix of all the local chains $\C^{(p)}(t)$ for $1 \leq p \leq n$. 
\end{itemize}

The process evolution is  as follows.

\begin{itemize}

    \item {\bf M0}: $\T(0)= \T^{(p)}(0) = \C^{(p)}(0), 1\leq p\leq n$ is a single root block, the genesis block.

    \item {\bf M1}: Adversary blocks are mined following a Poisson process at rate $\lambda_a$.  When a block is mined by the adversary, the mother tree $\T(t)$ is updated. The adversary can choose which block in $\T(t)$ to be the parent of the adversary block (i.e. the adversary can mine anywhere in the tree $\T(t)$.)
    
    \item {\bf M2}: Honest blocks are mined at a total rate of $\lambda_h$ across all the honest nodes, independent at each honest node and independent of the adversary mining process. When a block is mined by the honest node $p$, the sub-tree $\T^{(p)}(t)$ and the longest chain $\C^{(p)}(t)$ is updated. According to the longest chain rule, this honest block is appended to  the tip of $\C^{(p)}(t)$. The mother tree $\T(t)$ is updated accordingly.
    
    \item {\bf M3}:  $\T^{(p)}(t)$ and $\C^{(p)}(t)$ can also be updated by the adversary, in two ways: i) a block (whether is honest or adversary) must be added to $\T^{(p)}(t)$ within time $\Delta$ once it has appeared in $\T^{(q)}$ for some $q \neq p$, and the longest chain $\C^{(p)}(t)$ is extended if the block is at its tip; ii) the adversary can replace $\T^{(p)}(t^-)$ by another sub-tree $\T^{(p)}(t)$ from $\T(t)$ as long as the new sub-tree $\T^{(p)}(t)$ is an induced sub-tree of the new tree $\T^{(p)}(t)$, and can update $\C^{(p)}(t^-)$ to a longest chain in $T^{(p)}(t)$. \footnote{All jump processes are assumed to be right-continuous with left limits, so that $\C(t),\T(t)$ etc. include the new arrival if there is a new arrival at time $t$.} 
    
    \end{itemize}

We highlight the capabilities of the adversary in this model:

\begin{itemize}
\item {\bf A1}: Can choose to mine on any one block of the tree $\T(t)$ at any time.
\item {\bf A2}: Can delay the communication of blocks between the honest nodes, but no more than $\Delta$ time.
\item {\bf A3}: Can broadcast privately mined blocks at times of its own choosing: when private blocks are made public at time $t$ to node $p$, then these nodes are added to $\T^{(p)}(t^-)$ to obtain $\T^{(p)}(t)$. Note that by property M3(i), when private blocks appear in the view of some honest node $p$, they will also appear in the view of all other honest nodes by time $t+\Delta$. 
\item {\bf A4}: Can switch the $p$-th honest node's mining from one longest chain to another of equal length at any time, even when its view of the tree does not change. In this case, $\T^{(p)}(t) = \T^{(p)}(t^-)$ but $\C^{(p)}(t) \not = \C^{(p)}(t^-)$. 
\end{itemize}

The question is on what information can the adversary base in making these decisions? We will assume a causal adversary which has full knowledge of all past mining times of the honest blocks and the adversary blocks. 


Proving the security (persistence and liveness) of the protocol boils down to  providing a guarantee that the chain $\C(t)$ converges fast as $t \rightarrow \infty$ and that honest blocks enter regularly into $\C(t)$ regardless of the adversary's strategy. 


\subsection{From PoW to a unified model}

\begin{figure*}
    \centering
    \includegraphics[width=\textwidth]{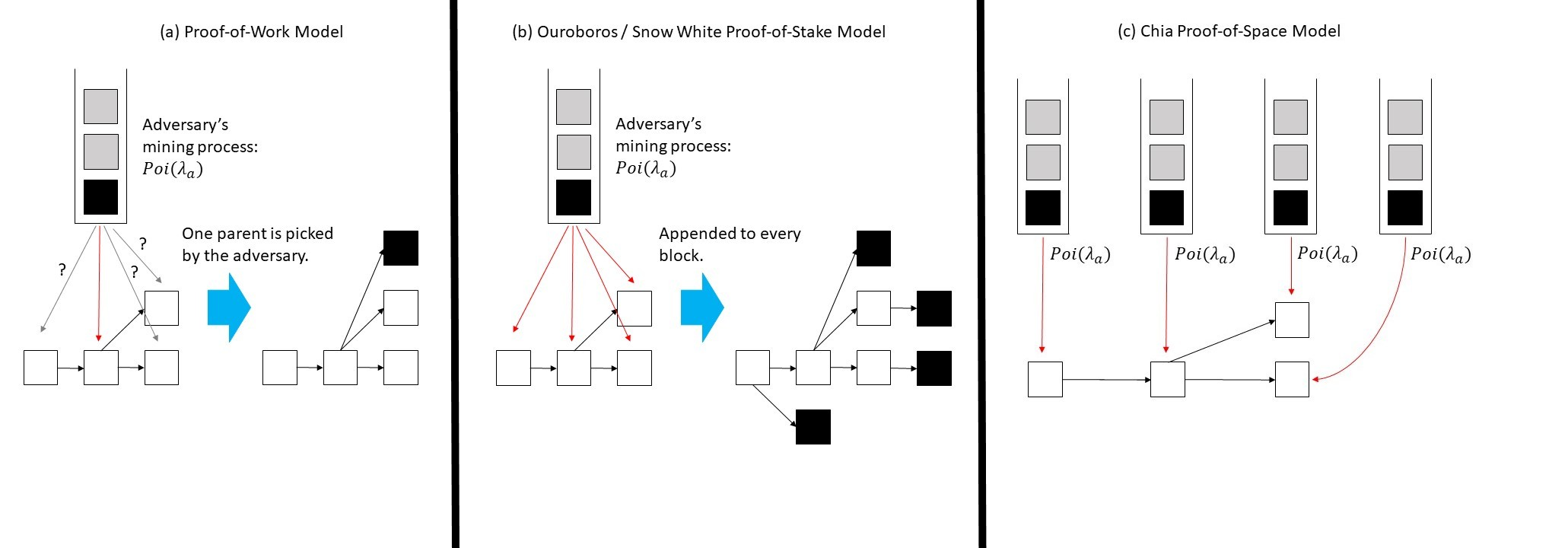}
    \caption{Three models for adversary block mining. In all models, adversary blocks are visualized as arriving via Poisson queues, and the focus is on how the block at the head of each queue is appended to the blocktree. In the PoW model, each adversary block can be appended to exactly one of the parent blocks of the existing blocktree. In the Paos/SnowWhte model, each adversary block can be appended to all possible parents blocks. In the Chia PoSpace model, the adversary blocks are mined independently on  the parent blocks of the existing tree.}
    \label{fig:models}
\end{figure*}

The model introduced in the last section can serve as a unified model for all three classes of protocols we study in this paper. The key difference between these classes of protocols is how the lottery in winning block proposal slots is conducted. This difference can be encapsulated by changing only one modeling assumption: {\bf M1}, the assumption on the adversary mining process (Figure \ref{fig:models}). In particular, the assumption on the honest behavior ({\bf M2}) remains the same,

\begin{itemize}
    \item {\bf M1-PoW} (Proof-of-Work): The original assumption we already had: Adversary blocks are mined according to a Poisson process 
    at rate $\lambda_a$, and the mined block can be appended to any parent block but only one, of the adversary's choosing, in the current mother tree $\T(t)$. This models the random attempts at solving the hash puzzle on one of the existing blocks.
    \item {\bf M1-PS} (Praos/SnowWhite Proof-of-Stake model): The adversary blocks are mined\footnote{In these Proof-of-Stake protocols, block proposal slots are won by conducting lotteries using the keys of the stake holders rather than by solving difficult computational puzzles as in Proof-of-Work protocols. However,for convenience, we use the term "mining" to denote the winning of any type of lotteries.}   according to a Poisson process at rate $\lambda_a$ (similar to PoW), but the adversary is allowed to append a version of each mined block simultaneously at {\em all} the blocks  in the current tree $\T(t)$. 
    \item {\bf M1-Chia} (Chia Proof-of-Space model): The adversary blocks are mined according to multiple independent Poisson processes of rate $\lambda_a$, one at each block of the current tree $\T(t)$. 
    A new block is appended to the tree at a certain block when a mining event happens.
\end{itemize}

Under {\bf M1-PoW}, miners can only mine on one parent block at a time, a consequence of conservation of work. Hence, the mined block can only be appended to one of the parent blocks. In {\bf M1-PS} and {\bf M1-Chia}, the adversary is able to mine new blocks on {\em all} of the existing blocks of the blocktree. This is a consequence of the phenomenon of {\em Nothing-at-stake}: the same resource (stake in PoS, disk space in PoSpace) can be used by the nodes to participate in random lotteries at all parent blocks to propose new blocks. Hence, unlike under assumption {\bf M1-PoW}, the overall mining rate of adversary blocks increases as the tree $\T(t)$ grows over time under both {\bf M1-PS} and {\bf M1-Chia}. However, the mining events across different blocks are fully {\em dependent} in {\bf M1-PS}  and completely {\em independent} in {\bf M1-Chia}. This is a consequence of the difference of how randomness is used in running the lotteries at different blocks. In the case of Praos/SnowWhite, the same randomness is used. In the case of Chia, independent randomness is used.

We note that it may appear that the capability {\bf A1} of the adversary (choosing where to mine), which is present in {\bf M1-PoW}, is gone under {\bf M1-PS} and {\bf M1-Chia}. However, the reason is that the adversary does not have to choose because it can mine everywhere simultaneously. Thus the adversary is actually more powerful under the {\bf M1-PS} and {\bf M1-Chia} conditions because the adversary has at its disposal much larger number of adversary blocks to attack the protocol. Somewhat surprisingly, our security threshold results show that this extra power is not useful in Praos/SnowWhite but useful in Chia.

The modeling assumptions for these protocols will be justified in more details in the following two subsections. The reader who is comfortable with these assumptions can go directly to Section \ref{sec:nak_blks}.

\input{app_pos}
\input{app_chia}

%% file: app_pos.tex
\subsection{Ouroboros Praos  and Snow White Proof-of-Stake model}
\label{app:ps}

   This section shows how   Ouroboros Praos \cite{david2018ouroboros} and Snow White \cite{bentov2016snow}  Proof-of-Stake protocols can be modeled using assumption {\bf M1-PS}  as mentioned earlier.  Both of these are Proof-of-Stake protocols, which means nodes get selected to create blocks in proportion to the number of coins (=stake) that they hold rather than the computation power held by the nodes. While the two protocols are similar at the level required for the analysis here, for concreteness, we will describe here the relation with Ouroboros Praos, which can handle adaptive corruption of nodes.
   
   We consider here only the static stake scenario - the stake of various nodes is fixed during the genesis block and assume that there is a single epoch (the composition of epochs into a dynamic stake protocol can be done using the original approach in \cite{david2018ouroboros}). The common randomness as well as the stake of various users is fixed at genesis (more generally, these are fixed at the beginning of each epoch for the entire duration of the epoch). For this protocol, we will assume that all nodes have a common clock (synchronous execution). At each time $t$, every node computes a verifiable random function (VRF) of the current time, the common randomness and the secret key. If the output value of the VRF is less than a certain threshold, then that node is allowed to propose a block at that time, to which it appends its signature. The key property of the VRF is that any node with knowledge only of the public key can validate that it was obtained with a node possessing the corresponding secret key. An honest node will follow the prescribed protocol and thus only create one block which it will append to the longest chain in its view. However, a winning dishonest node can create many different blocks mining on top of distinct chains. Blocks which are well-embedded into the longest-chain are considered confirmed.
   
   Now, we explain the connection of the protocol to our modeling in the earlier section. The first assumption is that time is quantized so finely that the continuous time modeling makes sense - this assumes that there is no simultaneous mining at any time point. However, if nodes mine blocks close to each other in time, they can be forked due to the delay $\Delta$ in the propagation time (thus we model concurrent mining through the effect of the propagation delay rather than through discrete time). Second, the honest action is to grow the longest chain through mining a new block at the tip - this justifies M2 (here $\lambda_h$ is proportional to the total honest stake). The adversaries can mine blocks which can be appended to many different positions in the blockchain. We assume that in the worst case, every adversary arrival contributes to a block extending every single block in the tree. We note that furthermore, there is another action, which is that the adversary can create many different blocks at any given position of the blockchain. Since this action does not increase the length of any chain or increase future mining opportunities, we do not need to model this explicitly. However, we point out that, since we show that a certain prefix of the blockchain ending at a honest block remains fixed for all future, that statement continues to hold even under this expanded adversary action space.


%% file: app_chia.tex
\subsection{Chia Proof-of-Space model}
\label{app:chia}

Chia consensus \cite{cohen2019chia} incorporates a combination of Proof of Space (PoSpace) and Proof of time, and is another energy efficient alternative to Bitcoin. PoSpace \cite{abusalah2017beyond,dziembowski2015proofs} is a cryptographic technique where provers can efficiently generate proofs to show that they allocate unused hard drive space for storage space.  Proof of time is implemented by a Verifiable Delay Function (VDF) \cite{boneh2018verifiable,pietrzak2018simple} that requires a certain amount of sequential computations to execute, but can be verified far quicker: a  VDF takes a challenge $c \in \{0,1\}^w$ and and a time parameter $t \in \mathbb{Z}^+$ as input, and produces a output $\tau$ and a proof $\pi$ in (not much more than) $t$ sequential steps; the correctness of output $\tau$ can be verified with the proof $\pi$ in  much less than $t$ steps. PoSpace enables Sybil resistance by  restricting participation to nodes that have reserved enough hard disk space and VDF enables coordination without having synchronized clocks as well as preventing long-range attacks \cite{park2018spacemint}. 
 
In Chia, each valid block $B$ contains a PoSpace $\sigma$ and a VDF output $\tau$. A Chia full node  mines a new block ($B_i$, with $i$ denoting the depth of the block from Genesis) as follows: 
\begin{enumerate}
    \item It first picks the  block $B_{i-1}$, at the tip of the longest chain in its local view of the blocktree, as the parent block that the newly generated block $B_i$ will be appended to.
    \item It draws a challenge $c_1$ deterministically from $B_{i-1}$ and generates a valid PoSpace $\sigma_i$ based on $c_1$ and a large file of size at least $M$ bits it stores.
    \item It computes a valid VDF output $\tau_i$ based on a challenge $c_2$ and a time parameter $t$, where $c_2$ is also drawn deterministically from $B_{i-1}$ and $t$ is the hash of $\sigma_i$ multiplied by a difficulty parameter $T$ (i.e. $t=0.{\sf H}(\sigma_i) \times T$ where {\sf H} is a cryptographic hash function).
    \item A new block $B_i$ comprised of $\sigma_i$, $\tau_i$ and some payload (example: transactions) is appended to $B_{i-1}$ in the blocktree.
\end{enumerate}

For each node, the ``mining'' time of a new block follows a uniform distribution in $(0,T)$: this is because  the hash function {\sf H} outputs a value that is uniformly distributed over its range. Suppose there are $N$ full nodes in the Chia network, then the inter-arrival block time in Chia consensus would be $\min(U_1,U_2,\cdots,U_N)$, where $U_i \sim {\sf Unif}(0,T)$ for $1 \leq i \leq N$. Then the expected inter-arrival block time is  
\begin{equation*}
\mathbb{E}[\min(U_1,U_2,\cdots,U_N)] = \int_0^T{(1-t/T)^N dt} = \frac{T}{N+1}.
\end{equation*}

So to maintain a fixed inter-arrival block time (example: 10 minutes in Bitcoin), the difficulty parameter $T$ needs to be adjusted linearly as number of full nodes $N$ grows.
We also observe that the chance for a node storing two large files each of size at least $M$ bits to find the first block is exactly doubled compared with a node storing one file, which provides Sybil resistance to Chia. 
Further we can model the mining process in Chia as a Poisson point process for large $N$. Fixing a parent block in the block tree, the number of new blocks mined in time $t$ follows a binomial distribution ${\sf bin}(N,t/T)$, which approaches a Poisson distribution ${\sf Poi}(Nt/T)$ when $N \rightarrow \infty$ and $N/T \rightarrow C$ for some constant $C$. 

 Assume there are $n$ honest nodes each controlling $M$ bits of space, and the adversary has $a \cdot M$ bits of space, then the mining processes of honest blocks and adversary blocks are Poisson point processes with rate $\lambda_h$ and $\lambda_a$ respectively, where $\lambda_h$ and $\lambda_a$ are proportional to total size of disk space controlled by honest nodes ($n \cdot M$) and the adversary ($a \cdot M$) respectively. Also while the honest nodes are following the longest chain rule, the adversary can work on multiple blocks or even all blocks in the block tree as a valid PoSpace is easy to generate and the adversary can compute an unlimited amount of VDF outputs in parallel; a similar phenomenon occurs in  Proof-of-Stake blockchains where it is termed as  the Nothing-at-Stake (NaS) attack  \cite{pos_paper}. Hence, we can model the adversary blocks as generated according to multiple independent Poisson processes of rate $\lambda_a$, one at each block of the current tree $\T(t)$. 
    A new block is appended to the tree at a certain block when a generation event happens. Like in the model for Ouroboros Praos and Snow White, the total rate of adversary block generation increases as the tree grows; however the generation events across different blocks are independent rather than fully dependent.

%% file: nakamoto_blocks.tex
\section{Blocktree partitioning and Nakamoto blocks}
\label{sec:nak_blks}

In this section, we will introduce the concept of {\em blocktree partitioning} to represent a general adversary attack as  a collection of adversary trees racing against a fictitious honest chain. Using this representation, we 
define the key notion of {\em Nakamoto blocks} as honest blocks that are the winners of the race against all the past trees, and show that if a block is a Nakamoto block, then the block will forever remain in the longest chain. The results in this section apply to all three models. In fact, they are valid for any assumption on the adversary mining process in {\bf M1} in the model in Section \ref{sec:model_nak}, because no statistical assumptions are made.  In Section \ref{sec:analysis}, we will perform security analysis in all three backbone models using the tool of Nakamoto blocks, by showing that they occur frequently with high probability whenever the adversary power is not sufficient to mount a successful private attack. This proves the liveness and persistency of the protocols. 

First, we introduce the concept of blocktree partitioning and define Nakamoto blocks in the simpler case when $\Delta = 0$, and then we extend to general $\Delta$. The unrealistic but pedagogically useful zero-delay case allows us to focus on the capability of the adversary to mine and publish blocks, while the general case brings in its capability to delay the delivery of blocks by the honest nodes as well.

\subsection{Network delay $\Delta = 0$}

\subsubsection{Blocktree partitioning}

Let $\tau^h_i$ and $\tau^a_i$ be the mining time of the $i$-th honest and adversary blocks respectively; $\tau^h_0 = 0$ is the mining time of the genesis block, which we consider as the $0$-th honest block. 

\begin{defn}
{\bf Blocktree partitioning}
Given the mother tree $\T(t)$, define for the $i$-th honest block $b_i$, the {\em adversary tree} $\T_i(t)$ to be the sub-tree of the mother tree $\T(t)$ rooted at $b_i$  and consists of all the adversary blocks that can be reached from $b_i$ without going through another honest block. The mother tree $\T(t)$ is partitioned into sub-trees $\T_0(t),\T_1(t), \ldots T_j(t)$, where the $j$-th honest block is the last honest block that was mined before time $t$.
\end{defn}
See Figure \ref{fig:tree_partition}(b) for an example.

The sub-tree $\T_i(t)$ is born at time $\tau^h_i$ as a single block $b_i$ and then grows each time an adversary block is appended to a chain of adversary blocks from $b_i$. Let $D_i(t)$ denote the depth of $\T_i(t)$; $D_i(\tau^h_i) = 0$.  

\subsubsection{Nakamoto blocks}
Let $A_h(t)$ be the number of honest blocks mined from time $0$ to $t$. $A_h(t)$ increases by $1$ at each time $\tau^h_i$. We make the following important definition.

\begin{defn}
\label{defn:nak_blk}
({\bf Nakamoto block for $\Delta = 0$}) Define 
\begin{equation}
\label{eqn:E}
    E^0_{ij} = \mbox{event that $D_i(t)  < A_h(t) - A_h(\tau^h_i)$ for all $t > \tau^h_j$}
\end{equation}
for some $i<j$. The $j$-th honest block is called a {\em Nakamoto block} if 
\begin{equation}
F^0_j = \bigcap_{i = 0}^{j-1} E^0_{ij}
\end{equation}
occurs.
\end{defn}

We can interpret the definition of a Nakamoto block in terms of a {\em fictitious system}, having the same block mining times as the actual system, where there is a growing chain consisting of only honest blocks and the adversary trees are racing against this honest chain.  (Figure \ref{fig:nakamoto_blocks}). The event $E^0_{ij}$ is the event that the adversary tree rooted at the $i$-th honest block does not catch with the fictitious honest chain {\em any} time after the mining of the $j$-th honest block.  When the fictitious honest chain reaches a Nakamoto block, it has won the race against {\em all} adversary trees rooted at the past honest blocks.


Even though the events are about a fictitious system with a purely honest chain and the longest chain in the actual system may consist of a mixture of adversary and honest blocks, the actual chain can only grow faster than the fictitious honest chain, and so we have the following key lemma showing that a Nakamoto block will stabilize and remain in the actual chain forever.


\begin{lemma}
\label{lem:regen} ({\bf Nakamoto blocks stabilize, $\Delta = 0$.})
If the $j$-th honest block is a Nakamoto block, then it will be in the longest chain $\C(t)$ for all $t > \tau^h_j$. Equivalently, $\C(\tau^h_j)$ will be a prefix of $\C(t)$ for all $t > \tau^h_j$.
\end{lemma}

\begin{proof}
Note that although honest nodes may have different views of the longest chain because of the adversary capability {\bf A4}, $\T^{(p)}(t) = \T^{(q)}(t)$ and hence $L^{(p)}(t) = L^{(q)}(t)$ always hold for any $q \neq p$ at any time $t$ when $\Delta = 0$.  Let $L(t)$ be the length of the longest chain in the view of honest nodes. $L(0) = 0$. Note that since the length of the chain $\C^{(p)}(t)$ increments by $1$ immediately at every honest block mining event (this is a consequence of $\Delta = 0$), it follows that for all $i$ and for all $t > \tau^h_i$,  
\begin{equation} 
    L(t)- L(\tau^h_i) \ge A_h(t) - A_h(\tau^h_i).
    \label{eq:lc}
\end{equation}
We now proceed to the proof of the lemma.

We will argue by contradiction. Suppose $F^0_j$ occurs and let $t^* > \tau^h_j$ be the smallest $t$ such that $\C(\tau^h_j)$ is not a prefix of $\C^{(p)}(t)$ for some $1\leq p \leq n$. Let $b_i$ be the last honest block on $\C^{(p)}(t^*)$ (which must exist, because the genesis block is by definition honest.) If $b_i$ is generated at some time $t_1> \tau^h_j$, then $\C^{(p)}(t_1^-)$ is the prefix of $\C^{(p)}(t^*)$ before block $b_i$, and does not contain $\C(\tau^h_j)$ as a prefix, contradicting the minimality of $t^*$. So $b_i$ must be generated before $\tau^h_j$, and hence $b_i$ is the $i$-th honest block for some $i < j$. The part of $\C^{(p)}(t^*)$ after block $b_i$ must lie entirely in the adversary tree $T_i(t^*)$ rooted at $b_i$. Hence,
\begin{equation*}
    L(t^*) \le L(\tau^h_i) + D_i(t^*).
\end{equation*}
However we know that
\begin{equation}
    D_i(t^*) < A_h(t^*)-A_h(\tau^h_i) \le L(t^*) - L(\tau^h_i),
\end{equation}
where the first inequality follows from the fact that $F_j$ holds, and the second inequality follows from the longest chain policy (eqn. (\ref{eq:lc})). From this we obtain that 
\begin{equation}
    L(\tau^h_i) + D_i(t^*) < L(t^*),
\end{equation}
which is a contradiction since $L(t^*) \le L(\tau^h_i) + D_i(t^*)$.
\end{proof}


Lemma \ref{lem:regen} justifies the name {\em Nakamoto block}: just like its namesake, a Nakamoto block has a godlike permanency. Also like its namesake, no one knows for sure whether a given block is a Nakamoto block: it is defined in terms of what happens in the indefinite future. However, the concept is useful because as long as a Nakamoto block appears in the last $k$ blocks of the current longest chain, then the prefix before these $k$ blocks will stabilize. Hence, the problem is reduced to showing under what conditions  Nakamoto blocks exist and they enter the blockchain frequently.

Since Nakamoto blocks are defined in terms of a race between adversary trees and the honest chain, and the growth rate of each adversary tree is bounded by the growth rate of the private attack adversary chain no matter what the attack is, one can intuitively expect that if the private attack is not successful, i.e. the growth rate of the private adversary chain is less than that of the honest chain, then once in a while Nakamoto blocks will occur because the adversary trees cannot win all the time. This intuition is made precise in Section \ref{sec:analysis} for the three models of interest. The current task at hand is to extend the notion of Nakamoto blocks to the $\Delta > 0 $ case.

\subsection{General network delay $\Delta$}

Definition \ref{defn:nak_blk} of a Nakamoto block is tailored for the zero network delay case. When the network delay $\Delta > 0$, there is forking in the blockchain even without adversary blocks, and two complexities arise:
\begin{enumerate}
    \item Even when a honest block $b$ has won the race against all the previous adversary trees, there can still be multiple honest blocks on the same level as $b$ in the mother tree $\T(t)$ due to forking. Hence there is no guarantee that $b$ will remain in the longest chain.
    \item Even when the honest block $b$ is the only block in its level, the condition in Equation \eqref{eqn:E} is not sufficient to guarantee the stabilization of $b$: the number of honest blocks mined is an over-estimation of the amount of growth in the honest chain due to forking.
\end{enumerate} 

The first complexity is a consequence of the fact that when the network delay is non-zero, the adversary has the additional power to delay delivery of honest blocks to create split view among the honest nodes. In the context of the formal security analysis of Nakamoto's PoW protocol, the limit of this power is quantified by the notion of {\em uniquely successful} rounds in \cite{backbone} in the lock-step synchronous round-by-round model, and extended to the notion of {\em convergence opportunities} in \cite{pss16} in the $\Delta$-synchronous model. The honest blocks encountering the convergence opportunities are called {\em loners}  in \cite{ren}.

\begin{defn}
The $j$-th honest block mined at time $\tau^h_j$ is called a {\em loner} if there are no other honest blocks mined in the time interval $[\tau^h_j - \Delta, \tau^h_j + \Delta]$.
\end{defn}

It is shown in \cite{pss16,ren} that a loner must be the only honest block in its depth in $\T(t)$ at any time $t$ after the block is mined.  Thus, to deal with the first complexity, we simply strengthen the definition of a Nakamoto block to restrict it to also be a loner block. Since loner blocks occur frequently, this is not an onerous restriction. 

To deal with the second complexity, we define the race of the adversary trees not against a fictitious honest chain without forking as in definition \ref{defn:nak_blk}, but against a fictitious honest tree with worst-case forking. This tree is defined as follows.

\begin{defn}
\label{def:fictitious}
Given honest block mining times $\tau^h_i$'s, define a honest fictitious tree $\T_h(t)$ as a tree which evolves as follows:
\begin{enumerate}
    \item $\T_h(0)$ is the genesis block.
    \item The first mined honest block and all honest blocks within $\Delta$ are all appended to the genesis block at their respective mining times to form the first level. 
    \item The next honest block mined and all honest blocks mined within time $\Delta$ of that are added to form the second level (which first level blocks are parents to which new blocks is immaterial) . 
    \item The process repeats.
\end{enumerate}
Let $D_h(t)$ be the depth of $\T_h(t)$.
\end{defn}

We are now ready to put everything together to define Nakamoto blocks in general.

\begin{defn}
\label{defn:nak_blk_gen}
({\bf Nakamoto block for general $\Delta$}) Let us define:
\begin{equation}
\label{eqn:Et}
    E_{ij} = \mbox{event that $D_i(t)  < D_h(t-\Delta) - D_h(\tau^h_i+\Delta)$ for all $t > \tau^h_j + \Delta$}.
\end{equation}
The $j$-th honest block is called a {\em Nakamoto block} if it is a loner and
\begin{equation}
F_j = \bigcap_{i = 0}^{j-1} E_{ij}
\end{equation}
occurs.
\end{defn}
Note that when $\Delta = 0$, $D_h(t) = A_h(t)$, the number of honest blocks mined in $[0,t]$. Hence $E_{ij} = E^0_{ij}$. Also, every block is a loner. Here Definition \ref{defn:nak_blk_gen} degenerates to Definition \ref{defn:nak_blk}. Moreover, it is not difficult to see that
$$D_h(t-\Delta) - D_h(\tau^h_i+\Delta) \le A_h(t) - A_h(\tau^h_i)$$
so Definition \ref{defn:nak_blk_gen} is indeed a strengthening of Definition \ref{defn:nak_blk}. This strengthening allows us to show that Nakamoto blocks stabilize for all $\Delta > 0$.

\begin{theorem}
\label{thm:nak_blk}({\bf Nakamoto blocks stabilize, general $\Delta$})
If the $j$-th honest block is a Nakamoto block, then it will be in the chain $\C(t)$ for all $t > \tau^h_j+\Delta$. This implies that the longest chain until the $j$-th honest block has stabilized.
\end{theorem}

The proof of Theorem \ref{thm:nak_blk} is given in \S \ref{sec:nakamoto-3.2}.

Nakamoto blocks are defined for general longest chain protocols. When applied to the Praos/SnowWhite protocols, the definition of Nakamoto blocks is a weakening of the definition of pivots in \cite{sleepy}. Although \cite{sleepy} did not define pivots explicitly in terms of races, one can re-interpret the definition as a race between the adversary and a fictitious honest chain consisting of {\em  only} loner honest blocks. This fictitious chain can never occur in the actual system even when no adversary blocks are made public, because there are other honest blocks which are not loners but can make it into the main chain. On the other hand, Nakamoto blocks are defined directly as a race between the adversary and the fictitious honest chain which would arise if there were no public adversary blocks. This is why the definition of Nakamoto blocks leads to a tight characterization of the security threshold in the Praos/SnowWhite protocols, matching the private attack threshold, while the definition of pivots in \cite{sleepy} cannot. (Theorem \ref{thm:pow_pos}). This tightening is similar to the tightening done in the recent work \cite{kiayias2020consistency} for the lock-step round-by-round model.








%% file: analysis.tex
\section{Security Analysis}
\label{sec:analysis}

The goal of this section is to show that the private attack is the worst attack for the three models defined in Section \ref{sec:model}. More precisely, we want to show that {\em security threshold}, i.e. the maximum adversary power tolerable for any adversary strategy, is the same as that of Nakamoto's private attack. This is true for any total mining rate $\lambda$ and for any $\Delta$. (In fact, the threshold depends only on the product $\lambda \Delta$.) We will use the notion of Nakamoto blocks to establish these results.

\input{results}

\input{approach}

\input{powpos_analysis}

\input{chia_analysis}

%% file: results.tex
\subsection{Statement of results}
\label{sec:results}

Our goal is to generate a transaction ledger that satisfies
{\em persistence}  and {\em liveness} as defined in  \cite{backbone}.
Together, persistence and liveness guarantee robust transaction ledger; honest transactions will be adopted to the ledger and be immutable.

\begin{definition}[from \cite{backbone}]
    \label{def:public_ledger}
    A protocol $\Pi$ maintains a robust public transaction ledger if it organizes the ledger as a blockchain of transactions and it satisfies the following two properties:
    \begin{itemize}
        \item (Persistence) Parameterized by $\tau \in \mathbb{R}$, if at a certain time a transaction {\sf tx} appears in a block which is mined more than $\tau$ time away from the mining time of the tip of the main chain of an honest node (such transaction will be called confirmed), then {\sf tx} will be confirmed by
        all honest nodes in the same position in the ledger.
        \item (Liveness) Parameterized by $u \in \mathbb{R}$, if a transaction {\sf tx} is received by all honest nodes for more than time $u$, then all honest nodes will contain {\sf tx} in the same place in the ledger forever.
    \end{itemize}
\end{definition}

As discussed in the introduction, the condition for the private attack on Nakamoto's Proof-of-Work protocol to be successful is
\begin{equation}
    \lambda_a > \lambda_{\rm growth}(\lambda_h,\Delta) = \frac{\lambda_h}{1 +\lambda_h \Delta}
\end{equation}
in the fully decentralized regime. In terms of $\beta$, the fraction of adversary power, and $\lambda$, the total block mining rate:
\begin{equation}
\label{eq:powpos_threshold}
\beta > \frac{1-\beta}{1 + (1-\beta)\lambda \Delta}.
\end{equation}
The parameter $\lambda \Delta$ is the number of blocks generated per network delay, and determines the latency and throughput of the blockchain. If this condition is satisfied, then clearly the ledger does not have persistency or liveness. Hence, the above condition can be interpreted as a tradeoff between latency/throughput and the security (under private attack).

In the Praos/SnowWhite protocol, the honest growth rate is the same as in the PoW system. Consider now the adversary blocks. They are mined according to a Poisson process at rate $\lambda_a$. When a block is mined, the adversary gets to append that block to all the blocks in the current adversary chain (cf.\ Figure \ref{fig:models}(b)). This leads to an exponential increase in the number of adversary blocks. However, the {\em depth} of that chain increases by exactly $1$. Hence the growth of the adversary chain is exactly the same as the advversary chain under PoW. Hence, we get exactly the same private attack threshold \eqref{eq:powpos_threshold} in both the PoW and the Praos/SnowWhite PoS protocols.

The theorem below shows that the  the private attack threshold yields the true security threshold for both classes of protocols.

\begin{theorem}
    \label{thm:pow_pos}
    If 
\begin{equation}
\beta < \frac{1-\beta}{1 + (1-\beta)\lambda \Delta},
\end{equation}
    then the Nakamoto PoW and the Ouroboros/SnowWhite PoS protocols generate  transaction ledgers such that each transaction ${\sf tx}$\footnote{In contrast to the theorems in \cite{backbone, pss16}, this theorem guarantees high probability persistence and liveness for {\em each} transaction rather than for the entire ledger. This is because our model has an infinite time-horizon while their model has a finite horizon, and guarantees for an infinite ledger is impossible. However, one can easily translate our results to high probability results for an entire finite ledger over a time horizon of duration polynomial in the security parameter $\sigma$ using the union bound.}
    satisfies {\em persistence} (parameterized by $\tau=\sigma$) and {\em liveness} (parameterized by $u=\sigma$) in Definition~\ref{def:public_ledger} with probability at least $1-e^{-\Omega(\sigma^{1-\epsilon})}$, for any $\epsilon > 0$.
\end{theorem}

For the Chia Proof-of-Space model, the private attack is analyzed in \cite{cohen2019chia, fan2018scalable}. The growth rate of the private adversary chain is $e\lambda_a$. (The magnification by a factor of $e$ is due to the Nothing-at-Stake nature of the protocol; more on that in Section \ref{sec:chia_analysis}.). Hence the condition for success for the private attack is:
\begin{equation}
    e \lambda_a > \frac{\lambda_h}{1 + \lambda_h \Delta},
\end{equation}
in the fully decentralized setting. This implies the following condition on $\beta$, the adversary fraction of space resources:
\begin{equation}
    e \beta >  \frac{1-\beta}{1 + (1-\beta)\lambda \Delta}.
\end{equation}

For the Chia model, this threshold yields the true threshold as well.
\begin{theorem}
    \label{thm:chia}
    If 
\begin{equation}
e \beta < \frac{1-\beta}{1 + (1-\beta)\lambda \Delta},
\end{equation}
    then the Chia Proof-of-Space protocol generate  transaction ledgers
    satisfying {\em persistence} (parameterized by $\tau=\sigma$) and {\em liveness} (parameterized by $u=\sigma$) in Definition~\ref{def:public_ledger} with probability at least $1-e^{-\Omega(\sigma^{1-\epsilon})}$, for any $\epsilon > 0$.
\end{theorem}

The security thresholds for the different models are plotted in Figure~\ref{fig:security_threshold}, comparing to existing lower bounds in the literature.

%% file: approach.tex
\subsection{Approach}

To prove Theorems \ref{thm:pow_pos} and \ref{thm:chia}, we use the technique of Nakamoto blocks developed in Section \ref{sec:nak_blks}. Theorem \ref{thm:nak_blk} states that Nakamoto blocks remain in the longest chain forever. The question is whether they exist and appear frequently regardless of the adversary strategy. If they do, then the protocol has liveness and persistency: honest transactions can enter the ledger frequently through the Nakamoto blocks, and once they enter, they remain at a fixed location in the ledger. More formally, we have the following result.

\begin{lemma}
\label{lem:nak_secure}
Define $B_{s,s+t}$ as the event that there is no Nakamoto blocks in the time interval $[s,s+t]$. If \begin{equation}
    P(B_{s,s+ t}) < q_t < 1
\end{equation}
for some $q_t$ independent of $s$ and the adversary strategy, then the protocol generates   transaction ledgers
    satisfying {\em persistence} (parameterized by $\tau=\sigma$) and {\em liveness} (parameterized by $u=\sigma$) in Definition~\ref{def:public_ledger} with probability at least $1-q_\sigma$.
\end{lemma}

The proof of Lemma \ref{lem:nak_secure} can be found in \S \ref{sec:persist_live}.
This reduces the problem to that of bounding the probability that there are no Nakamoto blocks in a long duration. Here we follow a similar style of reasoning as in the analysis of occurrence of pivots in the Sleepy Consensus protocol \cite{sleepy}:
\begin{enumerate}
    \item Show that the probability that the $j$-th honest block is a Nakamoto block is lower bounded by some $p >0$ for all $j$ and for all adversary strategy, in the parameter regime when the private attack growth rate is less than the honest chain growth rate.  
    \item Bootstrap from (1) to bound the probability of the event $B_{s,t}$, an event of no occurrence of Nakamoto blocks for a long time.
    
\end{enumerate}

Intuitively, if (1) holds, then one would expect that the chance that Nakamoto blocks do not occur over a long time is low, provided that a block being Nakaomoto is close to independent of another block being Nakamoto if the mining times of the two blocks are far apart. We perform the bootstrapping by exploiting this fact for the various models under consideration.

In \cite{sleepy}, the bootstrapping yields a bound $\exp(-\Omega(\sqrt{t}))$ on $P(B_{s,s+t})$. By recursively applying the bootstrapping procedure, we are able to get a bound $\exp(-\Omega(t^{1-\epsilon}))$ on $P(B_{s,s+t})$, for any $\epsilon> 0$. 
We apply this general analysis strategy to the three models in the next two subsections.

%% file: powpos_analysis.tex
\subsection{Nakamoto PoW and Praos/SnowWhite PoS Models}
\label{sec:powpos_analysis}

This subsection is dedicated to proving Theorem \ref{thm:pow_pos}. We will show that if 
\begin{equation}
    \lambda_a < \frac{\lambda_h}{1 + \lambda_h \Delta}, 
\end{equation}
then Nakamoto blocks occur frequently and regularly under both the PoW and the Praos/SnowWhite PoS models. Since the adversary in the Praos/SnowWhite PoS model is stronger, it suffices for us to prove the statement in that model.

As outlined in the section above, to prove Theorem \ref{thm:pow_pos}, we need to show that there exists constants $A_{\epsilon},a_{\epsilon}>0$ such that
$$P(B_{s,s+t}) < A_{\epsilon}\exp(-a_{\epsilon} t^{1-\epsilon})$$ for all $s,t>0$.
In this context, we first establish that the probability of occurrence of a Nakamoto block is bounded away from $0$. 

\begin{lemma}
\label{lem:infinite_many_F}
If 
$$\lambda_a < \frac{\lambda_h}{1+\Delta\lambda_h},$$
there exists a constant $p>0$ such that the probability that the $j$-th honest block is a Nakamoto block is  at least  $p$ for all $j$. 
\end{lemma}

The proof of Lemma \ref{lem:infinite_many_F} is given in \S \ref{sec:lem:infinite_many_F_2} of the Appendix. It is based on connecting the event of being a Nakamoto block to the event of a random walk never returning to the starting point. An alternative proof is presented in \S \ref{sec:lem:infinite_many_F} of the Appendix.

We next obtain a bound on $P(B_{s,s+t})$.

\begin{lemma}
\label{lem:pow-decay-delta-pos}
If 
$$\lambda_a < \frac{\lambda_h}{1+\Delta\lambda_h},$$
then for any $\epsilon > 0$
there exist constants $ a_\epsilon, A_\epsilon$ so that for all $s,t\geq 0$,
$$P(B_{s,s+t}) < A_{\epsilon}\exp(-a_{\epsilon} t^{1-\epsilon}).$$
\end{lemma}

Proof of Lemma \ref{lem:pow-decay-delta-pos} is given in \S \ref{sec:lem:pow-decay-delta-pos} of the Appendix. 
Then combining Lemma \ref{lem:pow-decay-delta-pos} with Lemma \ref{lem:nak_secure} implies Theorem \ref{thm:pow_pos}.

%% file: chia_analysis.tex
\subsection{Chia Proof-of-Space Model}
\label{sec:chia_analysis}

This subsection is dedicated to proving Theorem \ref{thm:chia}. We will show that if 
\begin{equation}
    e \lambda_a < \frac{\lambda_h}{1 + \lambda_h \Delta}, 
\end{equation}
then Nakamoto blocks occur frequently and regularly under the Chia Proof-of-Space model.

Since the occurrence of a Nakamoto block depends on whether the adversary trees from the previous honest blocks can catch up with the (fictitious) honest tree, we next turn to an analysis of the growth rate of an adversary tree. Note that under assumption ${\bf M1-Chia}$, adversary blocks are mined at rate $\lambda_a$ independently at each block of the mother tree $\T(t)$. Hence, each adversary tree $\T_i(t)$ grows statistically in the same way (and independent of each other). Without loss of generality, let us focus on the adversary tree $\tt_0(t)$, rooted at genesis, of the tree $\T(t)$. The depth of the tree $\tt_0(t)$ is $D_0(t)$ and defined as the maximum depth of its blocks. The genesis block is always at depth $0$ and hence $\tt_0(0)$ has depth zero.

With the machinery of branching random walks, we can show that the growth rate of depth of $\tt_0(t)$ is $e\lambda_a$ 
while the total number of adversary blocks in $\tt_0(t)$ 
grows exponentially with time $t$. Hence, compared to the Praos/SnowWhite model we just analyzed, the growth rate of each adversary tree is magnified by a factor of $e$. Thus, the Nothing-at-Stake phenomenon is more significant in the Chia model compared to the Praos/SnowWhite model, due to the independence of mining opportunities at different blocks.

We will also need  a tail bound on $D_0(t)$. While such estimates can be read from \cite{shi}, we bring instead a quantitative statement suited for our needs.
\begin{lemma}
   \label{theo-tail}
  For $m \geq 1$,
  \begin{equation}
    \label{eq:Ofertheo}
    P(D_0(t)\geq m)\leq   \left(\frac{e\lambda_a t}{m}\right)^m.
  \end{equation}
\end{lemma}

Details on the analysis of $\tt_0(t)$ and the proof of Lemma \ref{theo-tail} are in \S \ref{app:brw} in the Appendix. 

With Lemma \ref{theo-tail}, we show below that in the regime $e \lambda_a <  \frac{\lambda_h}{1 + \lambda_h\Delta}$, Nakamoto blocks has a non-zero probability of occurrence. 

\begin{lemma}
\label{lem:infinite_many_F_chia}
If 
$$e \lambda_a <  \frac{\lambda_h}{1 + \lambda_h\Delta},$$
then there is a $p > 0$ such that  that probability the $j$-th honest block is a Nakamoto block is greater than $p$ for all $j$. 

\end{lemma}

The proof of this result can be found in \S\ref{app:proof_delay} of the Appendix.

Having established the fact that Nakamoto blocks occurs with non-zero frequency, we can bootstrap on Lemma \ref{lem:infinite_many_F_chia} to get a bound on the probability that in a time interval $[s,s+t]$, there are no Nakamoto blocks, i.e. a bound on $P(B_{s,s+t})$.

\begin{lemma}
\label{lem:time_strong_chia}
If 
$$e \lambda_a <  \frac{\lambda_h}{1 + \lambda_h\Delta},$$
then for any $\epsilon > 0$
there exist constants $\bar a_\epsilon,\bar A_\epsilon$ so that for all $s,t\geq 0$,
\begin{equation}
\label{eqn:qst_strong}
P(B_{s,s+t}) \leq \bar A_\epsilon \exp(-\bar a_\epsilon t^{1-\epsilon}).
\end{equation}
\end{lemma}

The proof of this result can be found in \S\ref{app:proof_time_strong} of the Appendix.
Then combining Lemma \ref{lem:time_strong_chia} with Lemma \ref{lem:nak_secure} implies Theorem \ref{thm:chia}.

%% file: sample_path.tex
\section{Does Nakamoto really always win?}
\label{sec:worst_attack}

We have shown that the threshold for the adversary power beyond which the private attack succeeds is in fact the tight threshold for the security of the three models {\bf M1-PoW}, {\bf M1-PS} and {\bf M1-Chia}. However, security threshold is a statistical concept. Can we say that the private attack is the worst attack in a stronger, deterministic, sense?

Indeed, it turns out that one can, with a slight strengthening of the private attack, in a special case: the PoW model with network delay $\Delta = 0$. In this setting, we can indeed make a stronger statement.

In the PoW model, any attack strategy $\pi$ consists of two components: where to place each new adversary arrival and when to release the adversary blocks. Consider a specific attack $\pi_{\rm SZ}$: the Sompolinsky and Zohar's strategy of private attack with pre-mining \cite{zohar2016doublespend}. This attack focuses on a block $b$: it builds up a private chain with the maximum lead over the public chain when block $b$ is mined, and then starts a private attack from that lead. We have the following result.

\begin{theorem}
\label{lem:worst-attack}
Let $\tau_1^h, \tau_2^h, \ldots $ and $\tau_1^a, \tau_2^a, \ldots$ be a given sequence of mining times of the honest and adversary blocks. Let $b$ be a specific block. (i) Suppose $\pi$ violates the persistence of $b$ with parameter $k$, i.e. $b$ leaves the longest chain after becoming $k$-deep. Then the  $\pi_{SZ}$ attack on $b$ also forces $b$ to leave the longest chain after becoming $k$-deep, under the {\em same} mining times. (ii) Suppose $b$ is an honest block and $\pi$ violates liveness for the $k$ consecutive honest blocks starting with $b$, i.e. none of the $k$ consecutive honest blocks starting with $b$ stay in the longest chain indefinitely. Then the $\pi_{SZ}$ attack on $b$ also forces these $k$ consecutive honest blocks to leave the longest chain indefinitely under the {\em same} mining times.
\end{theorem}


The full proof of this theorem, together with a counter-example in the case of $\Delta > 0$, can be found in \S \ref{sec:lem:worst-attack}.
To demonstrate the main ideas used in the full proof, we focus here on a special case of where the adversary attacks the first honest block, $b_1$, mined after the genesis block.
Note that in this special case, the Sompolinsky and Zohar's attack strategy $\pi_{SZ}$ against $b_1$ is simply Nakamoto's private attack starting at the genesis block.
In this context, we prove that if the persistence of $b_1$ with parameter $k$ is violated by an adversary following some arbitrary attack strategy $\pi$, then, it is also violated by an adversary following the private attack under the same sequence of mining times for the honest and adversary blocks.
The proof will be built on the observation that at any depth, there can be at most one honest block when $\Delta=0$.
This observation is a direct result of the Chain Growth Lemma in \cite{backbone}, and is a consequence of the fact that there is no forking among the honest blocks when delay $\Delta = 0$.

\begin{proof}
Let $L(.)$ and $L^*(.)$ denote the lengths of the public longest chains, denoted by $\C$ and $\C^*$ under $\pi$ and the private attack respectively.
Let $\tau_1$ be the mining time of block $b_1$, and, define $t > \tau_1$ as the first time block $b$ disappears from $\C$ after it becomes $k$ deep within $\C$, under $\pi$.
Let $H$ and $A$ denote the number of honest and adversary blocks mined by time $t$ under the given sequence of mining times.


We first focus on $\pi$. Since $\pi$ removes $b_1$ from $\C$ at time $t$, there is another chain building on the genesis block that is parallel to $\C$ and at least as long as $\C$ at time $t$. (See top of Figure \ref{fig:priv_attack}.) Since there can be at most one honest block at every depth and there cannot be any honest block deeper than $L(t)$ (by virtue of the fact that $L(t)$ is the length of the public longest chain), $A \ge L(t) \ge H$. Also, since $b_1$ is at least $k$ deep at time $t$, $L(t) \geq k$. Hence, $A \ge \max\{H,k\}$.

\begin{figure}[h]
    \centering
    \vspace{-1.5mm}
    \includegraphics[width=0.89\linewidth]{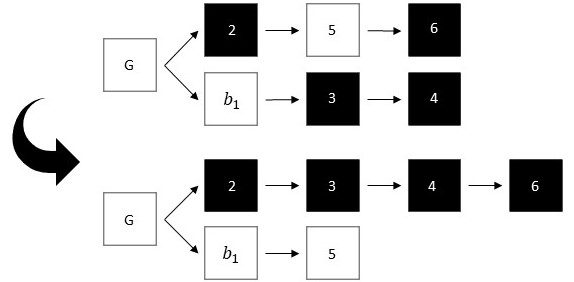}
    \vspace{-2mm}
    \caption{Blocktrees built under an arbitrary attack $\pi$ and the private attack by time $t$ are given at the top and bottom respectively.
    Colors black and white represent the adversary and the honest blocks, and, the blocks are labeled by the mining order. Here, $k =3, H = 2, A = 4, L(t) = 3, L^*(t) = 2.$ Under $\pi$, the adversary is successful in attacking $b_1$ at time $t$. Under the same mining times, the private attack has $4$ blocks in private and the honest chain has $2$ blocks. By the time the honest chain grows to $3$ blocks, the adversary can kick out $b_1$ by releasing the private chain.}  
    \vspace{-3mm}
    \label{fig:priv_attack}
\end{figure}

Now consider the blocktree under the private attack $\pi^*$ at time $t$. (Bottom of figure \ref{fig:priv_attack}.) Since no adversary block is mined on $\C^*$ under the private attack, $L^*(t)=H$. The length of the private chain starting at the genesis is exactly $A \geq \max\{H,k\} = \max\{L^*(t),k\}$.  If $L^*(t) \ge k$, the block $b_1$ can be kicked out now as the adversary can release the private chain at this time. On the other hand, if $L^*(t) <k$, the adversary can wait until the public chain grows to length $k$ and then release the private chain, which will be at least of length $k$. In either case, the private attack is successful in the violation of persistence for $b_1$. 

\end{proof}

In contrast to the PoW setting, a beautiful example from \cite{shi2019analysis} indicates that private attack is no longer the worst attack for every sequence of arrival times under the Praos/SnowWhite model, even for $\Delta = 0$. Figure \ref{fig:sample_path_1} explains this example, and exhibits the blocktree partitioning for this example. With only $1/3$ as many mining times opportunities to $2/3$ for the honest players, the protocol can lose persistence. A private attack would not be able to accomplish the same, because the adversary has less mining opportunities than the honest nodes.  This is somewhat surprising, given that the security threshold is $1/2$ for this model (at $\Delta = 0$). This also suggests that although the two settings, PoW and Praos/SnowWhite have identical security thresholds, their "true" worst case behaviors, taken over all mining time sequences, are different. The larger number of blocks available to the adversary in the Praos/SnowWhite protocol does have some effect in the true worst-case sense, and this allows the mounting of a more serious attack than a private attack. However, these are very atypical mining time sequences, and this difference does not show up in the security threshold.

So perhaps Nakamoto almost always wins.

\begin{figure}[h]
    \centering
    \vspace{-4mm}
    \includegraphics[width=\linewidth]{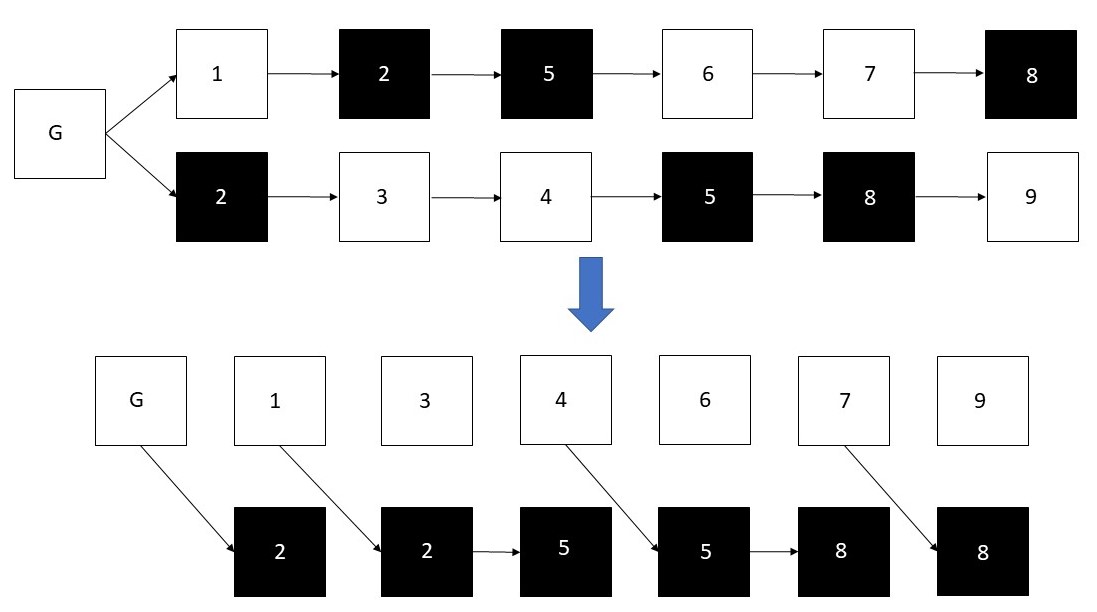}
    \vspace{-6mm}
    \caption{On the top is the blocktree for the example above. Colors black and white represent adversary and honest blocks respectively. The mining time of each block is stated on it. On the bottom is the partition of the blocktree into honest blocks and adversary chains, verifying that indeed there are no Nakamoto blocks. The adversary mines two blocks every third mining time and gets two copies of it. By publishing the shallower block and keeping the deeper block in private and having the honest nodes mine on the shallower block,  it can continue the balance attack indefinitely. This attacks relies on a periodic arrival pattern of the blocks. In a random environment, this pattern cannot hold indefinitely and the attack is not sustainable. So randomness saves Praos/SnowWhite.}
    \vspace{-4mm}
    \label{fig:sample_path_1}
\end{figure}

%% file: proofs.tex
\section{Definitions and Preliminary Lemmas for the Proofs}
\label{sec:defin}

In this section, we define some important events which will appear frequently in the analysis and provide some useful lemmas.

Let $\delta^h_i=\tau^h_{i}-\tau^h_{i-1}$ and $\delta^a_i=\tau^a_{i}-\tau^a_{i-1}$ denote the time intervals for subsequent honest and adversary arrival events. Let $d^h_i$ denote the depth of the $i$-th honest block within $D_h(t)$. Define $X_d$, $d>0$, as the time it takes for $D_h$ to reach depth $d$ after reaching depth $d-1$. In other words, $X_d$ is the difference between the times $t_1>t_2$, where $t_1$ is the minimum time $t$ such that $D_h(t)=d$, and, $t_2$ is the minimum time $t$ such that $D_h(t)=d-1$. 

Let $U_j$ be the event that the $j$-th honest block $b_j$ is a loner, i.e., 
$$U_j = \{\tau^h_{j-1} < \tau^h_{j} - \Delta\} \cap \{\tau^h_{j+1} > \tau^h_{j} + \Delta\}.$$
Let $\hat{F}_j = U_j \cap F_j$ be the event that $b_j$ is a Nakamoto block. Then we can define the following catch up event: 
\begin{equation}
\label{eqn:hat_Bik}
    \hat{B}_{ik} = \mbox{event that $D_i(\tau^h_{k} + \Delta) \ge  D_h(\tau^h_{k-1}) - D_h(\tau^h_i+\Delta)$},
\end{equation}
which is the event that the adversary launches a private attack starting from $b_i$ and catches up the fictitious honest chain right before $b_k$ is mined. The following lemma shows that event $\hat{F}_j$ can be represented with $\hat{B}_{ik}$'s.
\begin{lemma}  
\label{lem:F_j}
For each $j$,
\begin{equation}
    \hat{F}_j^c = F_j^c \cup U_j^c = \left(\bigcup_{(i,k): 0 \leq i < j <k} \hat{B}_{ik}\right) \cup U_j^c.
\end{equation}
\end{lemma}
\begin{proof}
\begin{eqnarray*}
&& U_j \cap E_{ij} \\
&=& U_j \cap \mbox{$\{D_i(t) < D_h(t-\Delta) - D_h(\tau^h_i+\Delta)$ for all $t > \tau^h_j + \Delta\}$}\\
&= & U_j \cap \mbox{$\{D_i(t +\Delta) < D_h(t) - D_h(\tau^h_i+\Delta)$ for all $t > \tau^h_j\}$ }\\
&= & U_j \cap \mbox{$\{D_i({\tau^h_k}^- +\Delta) < D_h({\tau^h_k}^-) - D_h(\tau^h_i+\Delta)$ for all $k > j\}$} \\
&= & U_j \cap \mbox{$\{D_i({\tau^h_k} +\Delta) < D_h({\tau^h_{k-1}}) - D_h(\tau^h_i+\Delta)$ for all $k > j\}$} 
\end{eqnarray*}
Since $\hat{F}_j = F_j \cap U_j = \bigcap_{0 \leq i < j} E_{ij} \cap U_j$, by the definition of $\hat{B}_{ik}$ we have $\hat{F}_j = \left(\bigcap_{(i,k): 0 \leq i < j <k} \hat{B}_{ik}^c\right) \cap U_j$. Taking complement on both side, we can conclude the proof.

Finally, define the parameter $r$ as follows:
$$r := \frac{\lambda_a}{\lambda_h} (1+\Delta \lambda_h),$$
for which $r<1$ holds whenever
$$\lambda_a < \frac{\lambda_h}{1+\Delta \lambda_h}.$$
\end{proof}

%% file: Appendix_Nakamoto.tex
\section{Proof of Theorem \ref{thm:nak_blk}}
\label{sec:nakamoto-3.2}

Notation used in this section is defined in section \ref{sec:nak_blks}. 

For the proof of the stabilization property of a Nakamoto block, it is crucial to show that $D_h(t)$ gives a conservative bound on the growth of the chains $\C^{(p)}$ from time $s$ to $t$. For this purpose, we prove the following proposition:

\begin{proposition}
\label{prop:delay}
For any given $s$, $t$ such that $s + \Delta < t - \Delta$;
$$D_h(t-\Delta) - D_h(s+\Delta) \leq L^{(p)}(t)-L^{(p)}(s)$$
for any honest miner $p$.
\end{proposition}

\begin{proof}
Assume that the increase in $L^{(p)}$ within the interval $[s,t]$ is solely due to the arrival of honest blocks to some miner in the interval $[s-\Delta,t]$. Then, we first show that delaying every block that arrives within this interval by $\Delta$ minimizes the increase in $L^{(p)}$ from $s$ to $t$ for any $t>s+\Delta$. To prove this, first observe that minimizing the increase in $L^{(p)}$ is equivalent to maximizing the time it takes for $\C^{(p)}$ to reach any depth $d$. Now, let $h_i$ be the block at the tip of $\C^{(p)}$ when it reaches depth $d$, and, assume that it took $\delta_i \leq \Delta$ time for $p$ to learn about $h_i$ after it was mined. Then, $\C^{(p)}$ reaches depth $d$ at time $\tau^h_i+\delta_i$. However, if the message for $h_i$ was delayed for $\delta'_i>\delta_i$ time, then, either $\C^{(p)}$ would have reached depth $d$ at time $\tau^h_i+\delta'_i \geq \tau^h_i+\delta_i$ with block $h_i$ at its tip, or, another block $h_j$, with index $j \neq i$ would have brought $\C^{(p)}$ to depth $d$ at some time $t$, $\tau^h_i+\delta'_i > t > \tau^h_i+\delta_i$. Hence, delaying the transmission of $h_i$ increases the time it takes for $\C^{(p)}$ to reach depth $d$. This implies that $h_i$ should be delayed as long as possible, which is $\Delta$. Since this argument also applies to any other block $h_j$ that might also bring $\C^{(p)}$ to depth $d$ when $h_i$ is delayed, every block should be delayed by $\Delta$ to maximize the time for $\C^{(p)}$ to reach any depth $d$. This, in turn, minimizes the increase in $L^{(p)}$ by any time $t>s$. 

Next, define the following random variable:
$$L_{max}(t)=\max_{p=1,..,n} (L^{(p)}(t)).$$
Then, we can assert that;
$$L_{max}(t-\Delta) \leq L^{(p)}(t) \leq L_{max}(t)$$
for any honest miner $p$. Then,
$$L^{(p)}(t)-L^{(p)}(s) \geq L_{max}(t-\Delta) - L_{max}(s).$$
From the paragraph above, we know that delaying every honest block by $\Delta$ minimizes $L^{(p)}(t)$ for any $t$. Hence, this action also minimizes $L^{(p)}(t)-L^{(p)}(s)$ for any $t>s+2\Delta$. Now, assume that no honest miner hears about any adversary block in the interval $[s,t]$ and every honest block is delayed by $\Delta$. Then, the difference $L_{max}(t-\Delta)-L_{max}(s)$ will be solely due the honest blocks that arrive within the period $[s,t-\Delta]$. However, in this case, depth of $L_{max}$ changes via the same process as $D_h$ (when each miner has infinitesimal power), which implies the following inequality:
$$L_{max}(t-\Delta)-L_{max}(s) \geq D_h(t-\Delta)-D_h(s+\Delta).$$
Hence, we see that when every block is delayed by $\Delta$ and there are no adversary blocks heard by $p$ in the time interval $[s,t]$;
$$L^{(p)}(t)-L^{(p)}(s) \geq D_h(t-\Delta)-D_h(s+\Delta).$$
However, delaying honest blocks less than $\Delta$ time or the arrival of adversary blocks to $p$ in the period $[s,t]$ only increases the difference $L^{(p)}(t)-L^{(p)}(s)$. Consequently;
$$D_h(t-\Delta)-D_h(s+\Delta) \leq L^{(p)}(t)-L^{(p)}(s)$$
for any honest miner $p$.

\end{proof}

Now we are ready to prove Theorem \ref{thm:nak_blk}.

\begin{proof}

We prove that the $j$-th honest block will be included in any future chain $\C(t)$ for $t>\tau^h_i+\Delta$, by contradiction. Suppose $\hat{F}_j$ occurs and let $t^* > \tau^h_j+\Delta$ be the smallest $t$ such that the $j$-th honest block is not contained in $\C^{(p)}(t)$ for some $1\leq p \leq n$. Let $h_i$ be the last honest block on $\C^{(p)}(t^*)$, which must exist, because the genesis block is by definition honest. If $\tau^h_i> \tau^h_j+\Delta$ for $h_i$, then, $\C^{(p)}(\tau^{h-}_i)$ is the prefix of $\C^{(p)}(t^*)$ before block $h_i$, and, does not contain the $j$-th honest block, contradicting the minimality of $t^*$. Therefore, $h_i$ must be mined before time $\tau^h_j+\Delta$. Since the $j$-th honest block is a loner, we further know that $h_i$ must be mined before time $\tau^h_j$, implying that $h_i$ is the $i$-th honest block for some $i < j$. In this case, part of $\C^{(p)}(t^*)$ after block $h_i$ must lie entirely in the tree $\tt_i(t^*)$ rooted at $h_i$. Hence, 
\begin{equation}
    L^{(p)}(t^*) \leq L^{(p)}(\tau^h_i) + D_i(t^*).
\end{equation}
However, we know that;
\begin{equation}
    D_i(t^*) < D_h(t^*-\Delta) - D_h(\tau^h_i+\Delta) \leq L^{(p)}(t^*)-L^{(p)}(\tau^h_i)
\end{equation}
where the first inequality follows from the fact that $\hat{F}_j$ holds and the second inequality follows from Proposition \ref{prop:delay}. From this we obtain that 
\begin{equation}
    L^{(p)}(\tau^h_i) + D_i(t^*) < L^{(p)}(t^*)
\end{equation}
which is a contradiction since $L^{(p)}(t^*) \leq L^{(p)}(\tau^h_i) + D_i(t^*)$. This concludes the proof.

\end{proof}

%% file: Appendix_PoW_Ouroboros.tex
\section{Proofs for Section \ref{sec:powpos_analysis}}
\label{sec:app_pow}

Notations used in this section are defined in \S \ref{sec:defin}. 

Subsequent propositions are used in future proofs.



\begin{proposition}
\label{prop:min-rate}
Let $Y_d$, $d \geq 1$, be i.i.d random variables, exponentially distributed with rate $\lambda_h$. Then, each random variable $X_d$ can be expressed as $\Delta+Y_d$.
\end{proposition}

\begin{proof}
Let $h_i$ be the first block that comes to some depth $d-1$ within $\T_h$. Then, every honest block that arrives within the interval $[\tau^h_i, \tau^h_i+\Delta]$ will be mapped to the same depth as $h_i$, i.e $d-1$. Hence, $\T_h$ will reach depth $d$ only when an honest block arrives after time $\tau^h_i+\Delta$. Now, we know that the difference between $\tau^h_i+\Delta$ and the arrival time of the first block after $\tau^h_i+\Delta$ is exponentially distributed with rate $\lambda_h$ due to the memoryless property of the exponential distribution. This implies that for each depth $d$, $X_d$ is equal to $\Delta+Y_d$ for some random variable $Y_d$ such that $Y_d$, $d \geq 1$, are i.i.d and exponentially distributed with rate $\lambda_h$. Then, $X_d$ are also i.i.d random variables with mean
$$\Delta + \frac{1}{\lambda_h}.$$

\end{proof}

\begin{proposition}
\label{prop:bound-1}
For any constant $a$, 
$$P(\sum_{d=a}^{n+a} X_d > n(\Delta + \frac{1}{\lambda_h})(1+\delta)) \leq e^{-n\Omega(\delta^2(1+\Delta\lambda_h)^2)}.$$
\end{proposition}

Proposition \ref{prop:bound-1} is proven using a Chernoff bound analysis and Proposition \ref{prop:min-rate}.

\begin{proposition}
\label{prop:bound-2}
Probability that there are less than
$$n\frac{\lambda_a(1-\delta)}{\lambda_h}$$ 
adversary arrival events from time $\tau^h_0$ to $\tau^h_{n+1}$ is upper bounded by
$$e^{-n\Omega(\delta^2\frac{\lambda_a}{\lambda_h})}.$$
\end{proposition}

Proposition \ref{prop:bound-2} is proven using the Poisson tail bounds.

\begin{proposition}
\label{prop:bound-3}

Define $B_n$ as the event that there are at least $n$ adversary arrivals while $D_h$ grows from depth $0$ to $n$:
$$B_n = \{\sum_{i=1}^n X_i \geq \sum_{i=0}^n \delta^a_i \}$$
If 
$$\lambda_a < \frac{\lambda_h}{1+\lambda_h \Delta},$$
then,
$$P(B_n) \leq e^{- A_0 n},$$
where,
$$A_0 = s\Delta + \ln(\frac{\lambda_a\lambda_h}{(\lambda_h-s)(\lambda_a+s)})>0$$
and,
$$s = \frac{\lambda_h - \lambda_a}{2}+\frac{2-\sqrt{4+\Delta^2(\lambda_a+\lambda_h)^2}}{2\Delta}.$$
\end{proposition}

Proof is by using Chernoff bound, and, optimizing for the value of $s$. It also uses Proposition \ref{prop:min-rate}.


\subsection{Proof of Lemma \ref{lem:infinite_many_F}}
\label{sec:lem:infinite_many_F_2}

The proof is based on random walk theory.

\begin{proof}
We would like to lower bound the probability that the $j$-th honest block is a loner and $F_j$ happens. Since the $j$-th honest block is a loner with probability $e^{-2\lambda_h\Delta}>0$ for all $j$, the probability that it is a Nakamoto block can be expressed as
$$P(F_j\ |\ \text{j-th honest block is a loner})\cdot e^{-2\lambda_h\Delta}$$
Then, the proof is reduced to obtaining a lower bound on 
$$P(F_j\ |\ \text{j-th honest block is a loner}).$$ 
For this purpose, we assume that the $j$-th honest block is a loner, and, proceed to obtain a lower bound on the probability of the event $F_j$: 

For any adversary tree $\tt_i$, $i<j$;
$$D_i(t)<D_h(t-\Delta)-D_h(\tau^h_i+\Delta)$$
for all times $t>\tau^h_j+\Delta$, which is equivalent to 
$$D_i(t+\Delta)<D_h(t)-D_h(\tau^h_i+\Delta)$$
for all times $t>\tau^h_j$. 

Let $U_j$ be the event that the $j$-th honest block is a loner. Let $G_j$ be the event that no adversary block is mined within the time period $[\tau^h_j,\tau^h_j+\Delta]$. Then, $P(G_j) = e^{-\lambda_a\Delta}$, and, we can lower bound $P(F_j|U_j)$ in the following way:
$$P(F_j|U_j) \geq P(F_j \cap G_j|U_j) = e^{-\lambda_a\Delta} P(F_j|U_j,G_j)$$
Since the events $G_j$, $j=1,2,..$ are shift invariant, the probability $P(F_j|U_j,G_j)$ is equal to the probability of the following event $\hat{F}_j$:

For any adversary tree $\tt_i$, $i<j$;
$$D_i(t)<D_h(t)-D_h(\tau^h_i+\Delta)$$
for all times $t>\tau^h_j$. Now, define $D^*(t)$ as the depth of the deepest adversary tree at time $t$ for $t\geq \tau^h_j$: 
$$D^*(t) := \max_{0\leq i < j} D_i(t)+D_h(\tau^h_i+\Delta)$$
Then, $\hat{F}_j$ basically represents the event that $D^*$ is behind $D_h$ for all times $t\geq \tau^h_j$. 
\newline

We next express $\hat{F}_j$ in terms of the following events:
$$E_1 := \{D^*(\tau^h_j) < D_h(\tau^h_j)\}$$
$E_1$ is the event that the tip of the deepest adversary tree, $D^*$, is behind the tip of the honest tree, $D_h$ at the arrival time of the $j$-th honest block.

$E_2$ is the event that $D_h(t)-D_h(\tau^h_j)$ is greater than the number of adversary arrivals during the time period $[\tau^h_j,t]$ for all $t$, $t>\tau^h_j$.

$E_3$ is the event that $D_h(\tau^h_j)-D_h(\tau^h_i+\Delta)$ is greater than the number of adversary arrivals during the time period $[\tau^h_i,\tau^h_j]$ for all $i$, $0 \leq i<j$. 

We can now express $\hat{F}_j$ in terms of $E_1$ and $E_2$: 
$$E_1 \cap E_2 \subseteq \hat{F}_j$$
Moreover, when a new adversary block is mined, depth of any of the trees $\tt_i$, $i<j$, increases by at most $1$. Hence, $E_3$ implies that none of the trees $\tt_i$, $i<j$, has depth greater than or equal to $D_h(\tau^h_j)$ at time $\tau^h_j$. Consequently, 
$$E_3 \subseteq E_1,$$
which further implies
$$E_3 \cap E_2 \subseteq \hat{F_j}$$
Observing that $E_3$ and $E_2$ are independent events, we can express the probability of $\hat{F_j}$ as;
$$P(\hat{F_j}) \geq P(E_3) P(E_2)$$

Now, define $E'_2$ as the event that $D_h(t-\Delta)-D_h(\tau^h_j)$ is greater than the number of adversary arrivals during the time period $[\tau^h_j,t]$ for all $t$, $t>\tau^h_j+\Delta$. Let $G'_j$ be the event that there is no adversary arrival during the time interval $[\tau^h_j,\tau^h_j+\Delta]$. Observe that again, $P(G'_j)=e^{-\lambda_a\Delta}$, and, the events $G'_j$, $j=1,2,..$ are shift invariant. Hence, we can do a similar trick as was done for the probabilities of $F_j$ and $G_j$ to obtain
$$P(E'_2) \geq e^{-\lambda_a\Delta} P(E_2).$$

Since the increase times of $D_h$ and the inter-arrival times of adversary arrivals are i.i.d, the growth processes of $D_h$ and the number of adversary blocks are time reversible. Hence, probability of $E_3$ approaches that of $E'_2$ from above as $j \to \infty$. Then, for all $j$, we can write
$$P(\hat{F_j}) \geq P(E_3) P(E_2) \geq P(E'_2)P(E_2) \geq e^{-\lambda_a\Delta} P(E_2)^2 $$
\newline

We now calculate the probability of the event $E_2$. To aid us in the calculation of $P(E_2)$, we construct a random walk $S[n]$. Here, the random walk is parametrized by the total number of adversary arrivals and increases in $D_h$ since time $\tau^h_j$. $S[n]$ stands for the difference between the increase in $D_h$ and the number of adversary arrivals when there has been, in total, $n$ number of increases in $D_h$ or adversary arrivals since time $\tau^h_j$. Notice that when $\Delta=0$, $D_h$ increases by one whenever there is an honest arrival. Hence, $S[n]$ simply counts the difference between the number of honest and adversary arrivals when there are $n$ arrivals in total. In this case, $S[n]$ jumps up by $1$ when there is an honest arrival, and, goes down by $1$ when there is an adversary arrival. Since the event that whether the next arrival is honest or adversary is independent of the past arrivals, $S[n]$ is a random walk when $\Delta=0$.
\newline

On the other hand, when $\Delta>0$, we have to construct a slight different random walk $S[n]$ for the difference between the increase in $D_h$ and the number of adversary arrivals due to the $\Delta$ dependence. Although this random walk has non-intuitive distributions for the jumps, we observe that 
\begin{enumerate}
    \item Expectation of these jumps is positive as long as 
    $$\lambda_a < \frac{\lambda_h}{1+\lambda_h\Delta}$$
    \item Expectation of the absolute value of the jumps is finite.
\end{enumerate}   
Then, due to the Strong Law of Large Numbers, every state of this random walk is transient, and, the random walk has a positive drift. This implies that starting at $S[0]=1$, the probability of $S[n]$ hitting or falling below $0$ is equal to some number $1-c$, where $1 \geq c>0$. 
\newline

Finally, observe that the probability of $S[n]$ hitting or falling below $0$ is exactly the probability of the event $E^c_2$. Hence, $P(E_2)=c>0$. Combining this observation with previous findings yields the following lower bound for $P(F_j|U_j)$: 
$$P(F_j|U_j) \geq e^{-\lambda_a\Delta} P(F_j|U_j,G_j) \geq e^{-2\lambda_a\Delta} P(E_2)^2 = e^{-2\lambda_a\Delta} c^2 = p > 0 $$
where $p>0$ does not depend on $j$. This concludes the proof.

\end{proof}

\subsection{Alternative Proof of Lemma \ref{lem:infinite_many_F}}
\label{sec:lem:infinite_many_F}

In this subsection, we give an alternative proof of Lemma \ref{lem:infinite_many_F}. We first present a proof sketch below:

First, in the Praos/SnowWhite model, each arrival of an adversary block can increment the depth of each adversary tree by exactly one although many copies of this block are mined. Hence, the adversary trees grow with the same rate as the adversary's mining rate, namely $\lambda_a$. Second, we observe that in the long run, the honest tree $T_h$ grows with rate $\lambda_h/(1+\lambda_h \Delta)$. See Figure \ref{fig:nakamoto_blocks} for the relation between adversary trees and the honest tree. Then, if $\lambda_a<\lambda_h/(1+\lambda_h \Delta)$, the gap between the depths of the honest tree and an adversary tree is expected to increase over time. This implies that the adversary trees built on blocks far from the tip of the honest tree will fall behind and not be able to catch-up with the honest tree. Hence, in order to analyze the probability that the $j$-th honest block $h_j$ is a Nakamoto block for any $j$, it is sufficient to focus on the adversary trees that have been built on honest blocks immediately preceding $h_j$. However, as there is only a small number of honest blocks immediately preceding $h_j$, there is a non-zero probability that none of the adversary trees built on them will be able to catch-up with the honest tree. Hence, when $\lambda_a < \lambda_h/(1+\lambda_h \Delta)$, $h_j$ becomes a Nakamoto block with non-zero probability for any $j$. 

We now proceed with a complete proof of lemma \ref{lem:infinite_many_F}.

\begin{proof}

We first observe that $(F_j \cap U_j)^c = F^c_j \cup U^c_j$ can be expressed as the union of the following disjoint events: (i) $h_j$ is not a loner. (ii) $h_j$ is a loner, and, $F^c_j$ happens:
\begin{IEEEeqnarray}{rCl}
\label{eq:p1}
P(F^c_j) = P(U^c_j) + P(F^c_j|U_j) P(U_j)
\end{IEEEeqnarray}
Now, since there exists a constant $c_1$ such that $0<c_1 \leq P(U_j)$ for any $j$, to prove the lemma, it is sufficient to find an upper bound $c_2$ on $P(F^c_j|U_j)$ such that $P(F^c_j|U_j) \leq c_2<1$. Hence, from now on, we assume that $h_j$ is a loner and given this fact, analyze the event $F^c_j$.

Note that the `catch-up' event $F^c_j$ implies the existence of a {\em minimum} time $t_j \geq \tau^h_j + \Delta$ such that there exists an adversary tree $\tt_i$ extending from some honest block $h_i$, $i<j$, for which, the following holds: 
$$D_i(t_j) \geq D_h(t_j-\Delta)-D_h(\tau^h_i+\Delta).$$
Define $l_j=D_h(\tau^h_i)+1$, and, $r_j = D_h(t_j-\Delta)$. Since $D_h$ is monotonically increasing and $t_j \geq \tau^h_j + \Delta$, and, $D(\tau^h_i) < D(\tau^h_j)$ as $h_j$ is a loner; $r_j = D_h(t_j-\Delta) \geq D_h(\tau^h_j) = d^h_j$, and, $l_j = D_h(\tau^h_i)+1 \leq D_h(\tau^h_j) = d^h_j$. Combining these facts, we infer that $l_j \leq d^h_j \leq r_j$. Consequently, at each depth within the interval $[l_j,r_j]$, there exists at least two blocks, at least one of which is an honest block in $\T_h$ and one of which is an adversary block in $\tt_i$. 

We next focus on the depth interval $[l_j,r_j]$, and, fix some constant and large integer $B$. Let $B_1$ and $B_2$ be the disjoint events that the `catch-up' event happens before and after depth $d^h_j+B$ respectively:
$$B_1 = \{r_j < d^h_j+B\}$$
$$B_2 = \{r_j \geq d^h_j+B\}$$
Let $B'_2$ denote the event that there exists a time $t$ such that $D_i(t) \geq D_h(t-\Delta)-D_h(\tau^h_i+\Delta)$ and $D_h(t-\Delta) \geq d^h_j+B$ for this $t$. Let $r'_j := D_h(t-\Delta)$. Then, we observe that
$$B'_2 \cap B^c_1 = B_2$$
Using this, we can upper bound $P(F^c_j|U_j)$ in the following way:
\begin{IEEEeqnarray}{rCl}
P(F^c_j|U_j) &=& P(B_1) + P(B_2) \\
&=& P(B_1) + P(B'_2 \cap B^c_1) \\
&=& P(B_1) + P(B'_2 | B^c_1) P(B^c_1) \\
&=& P(B_1) + P(B'_2 | B^c_1) (1-P(B_1)) \\
&\leq& P(B_1) + P(B'_2) (1-P(B_1))
\end{IEEEeqnarray}
Assume that $P(B_1)<1$ for all $j$. We will later prove that this is indeed true. Moreover, note that conditioning on $B^c_1$ decreases the probability of the event $B'_2$ since, (i) $B^c_1$ requires $\tt_i$ to be behind $\T_h$ while $\T_h$ increases through depths $d^h_j$ to $d^h_j+B$, (ii) Given any initial depths for $\tt_i$ and $\T_h$, catch-up events are ergodic. Then, proving that $P(F^c_j|U_j) \leq c_2$ for some $c_2<1$ reduces to proving that $P(B'_2) \leq c_4$ for some $c_4<1$ for a fixed $B$.

We next calculate an upper bound for $P(B'_2)$. For this purpose, we first define $B'_{a,b}$ as the event that at least $b-a$ adversary arrival events have occurred during the time interval $[\sum_{n=0}^a X_n-\Delta, \sum_{n=0}^b X_n+\Delta]$. Via the ergodicity of the processes $X_n$, using Proposition \ref{prop:bound-3}, we can write the following upper bound for $P(B'_{a,b})$ for $b-a$ sufficiently large:
$$P(B'_{a,b}) = P(\sum_{n=0}^{b-a} X_n+2\Delta \geq \sum_{n=0}^{b-a} \delta^a_n) \leq A_1 e^{- A_0 (b-a)}$$
where $A_1$ is a constant that is a function of $\Delta$, and, $A_0$ is the constant given in Proposition \ref{prop:bound-3}. Now, $B'_2$ requires at least $r'_j-l_j \geq B$ adversary blocks to be mined at the tip of $\tt_i$ from time $\tau^h_i+\Delta$ until some time $t-\Delta$, during which the depth of $\T_h$, $D_h$, grows by exactly $r'_j-l_j$. However, this is only possible if the adversary has at least $r'_j-l_j$ arrival events during this time interval. Hence, we can express $B'_2$ as a subset of the union of the events $B'_{a,b}$ in the following way:
$$B'_2 \subseteq \bigcup_{a < d^h_j, b \geq d^h_j+B} B'_{a,b}$$
Then, via union bound, its probability is upper bounded as shown below:
\begin{IEEEeqnarray}{rCl}
\label{eq:p3}
P(B'_2) &\leq& \sum_{0 \leq a < d^h_j} \sum_{b \geq d^h_j + B} P(B_{a,b}) \\
&\leq& A_1 e^{- A_0 B} \sum_{a=0}^{\infty} \sum_{b=0}^{\infty} e^{-A_0 (a+b)} \\
&=& A_1 \frac{1}{(1-e^{-A_0})^2} e^{- A_0 B} \\
&<& c_4
\end{IEEEeqnarray}
for sufficiently large $B$ and any $j$. Here, any positive constant $c_4$ smaller than $1$ can be achieved by making $B$ large enough.

Finally, for the $B$ fixed above, we prove that there exists a constant $c_3<1$ such that $P(B_1)\leq c_3$ for all $j$. Note that if $l_j \geq d^h_j-B$, there is a non-zero probability that no adversary block is mined from time $\tau^h_i$ to the time $B_h$ reaches depth $d^h_j+B$. Then there exists a constant $c_{31}<1$ such that $P(B_1|l_j \geq d^h_j-B)\leq c_{31}$ for all $j$. On the other hand, if $l_j < d^h_j-B$, we know from the calculations above that 
$$P(B_1|l_j < d^h_j-B) \leq A_1\frac{1}{(1-e^{-A_0})^2}e^{-A_0 B}$$
for all $j$. We further know that, for the large $B$ fixed above, there exists a constant $c_4<1$ such that this expression is below $c_4$. Consequently, for any $j$,
\begin{IEEEeqnarray}{rCl}
P(B_1) &\leq& P(l_j \geq d^h_j-B) c_{31} + (1-P(l_j \geq d^h_j-B))c_4 \\
&\leq& \max(c_{31},c_4) = c_3 < 1
\end{IEEEeqnarray}
Then, for any $j$, 
$$P(F^c_j|U_j) \leq c_3 + c_4 (1-c_3) < 1$$
This concludes the proof.

\end{proof}

\subsection{Proof of Lemma \ref{lem:pow-decay-delta-pos}}
\label{sec:lem:pow-decay-delta-pos}

We first state the following lemma which will be used in the proof of Lemma \ref{lem:pow-decay-delta-pos}. Recall that we have defined event $\hat{B}_{ik}$ in \S \ref{sec:defin} as:
\begin{equation}
    \hat{B}_{ik} = \mbox{event that $D_i(\tau^h_k+\Delta) \ge  D_h(\tau^h_{k-1}) - D_h(\tau^h_i+\Delta)$}.
\end{equation}

\begin{lemma}
\label{lem:PBik_pow}
There exists a constant $c>0$ such that
$$P(\hat{B}_{ik}) \leq e^{-c (k-i-1)}$$
\end{lemma}

\begin{proof}
We know from Proposition \ref{prop:bound-2} that there are more than $(1-\delta)(k-i)\lambda_a/\lambda_h$ adversary arrival events in the time period $[\tau^h_i,\tau^h_{k}+\Delta]$ except with probability $e^{-\Omega((k-i)\delta^2\lambda_a/\lambda_h)}$. Moreover, Proposition \ref{prop:bound-3} states that
$$P(\sum_{i=1}^n X_i \geq \sum_{i=0}^n \delta^a_i) \leq e^{-A_0 n}$$
for large $n$. Then, using the union bound, we observe that for any fixed $\delta$, probability of $\hat{B}_{ik}$ when there are more than $(1-\delta)(k-i)\lambda_a/\lambda_h$ adversary arrival events in the time period $[\tau^h_i,\tau^h_{k}+\Delta]$ is upper bounded by the following expression:
$$\frac{1}{1-e^{-C_1}} e^{-C_1(k-i)}$$
where 
$$C_1 = \frac{A_0(1-\delta)\lambda_a}{\lambda_h}.$$
Hence,
$$P(\hat{B}_{ik}) < \frac{1}{1-e^{-C_1}} e^{-C_1(k-i)} + e^{-\Omega((k-i)\delta^2\frac{\lambda_a}{\lambda_h})} \leq C_2 e^{-C_3(k-i)}$$
for any $k,i$, $k>i+1$, and appropriately chosen constants $C_2,C_3>0$ as functions of the fixed $\delta$. Finally, since $P(\hat{B}_{ik})$ decreases as $k-i$ grows and is smaller than $1$ for all $k>i+1$, we obtain the desired inequality for a sufficiently small $c\leq C_3$.

\end{proof}

We can now proceed with the main proof.

We divide the proof in to two steps. In the first step, we prove for $\varepsilon = 1/2$. By Lemma \ref{lem:F_j}, we have
\begin{equation}
    \hat{F}_j^c = F_j^c \cup U_j^c = \left(\bigcup_{(i,k): i < j <k} \hat{B}_{ik}\right) \cup U_j^c.
\end{equation}
Divide $[s,s+t]$ into $\sqrt{t}$ sub-intervals of length $\sqrt{t}$, so that the $r$ th sub-interval is:
$$\J_r : = [s+  (r-1) \sqrt{t}, s+ r\sqrt{t}].$$

Now look at the first, fourth, seventh, etc sub-intervals, i.e. all the $r = 1 \mod 3$ sub-intervals. Introduce the event that in the $\ell$-th $1 \mod 3$th sub-interval,
an adversary tree that is rooted at a honest block arriving in that sub-interval or in the previous ($0 \mod 3$) sub-interval catches up with a honest block in that sub-interval or in the next ($2 \mod 3$) sub-interval. 
Formally,
$$C_{\ell}=\bigcap_{j: \tau^h_j \in \J_{3\ell+1}}
U_j^c \cup \left(\bigcup_{(i,k): \tau^h_j - \sqrt{t} < \tau^h_i < \tau^h_j, \tau^h_j < \tau^h_k +\Delta < \tau^h_j +\sqrt{t} } \hat{B}_{ik} \right).$$
Note that for distinct $\ell$, the events $C_\ell$'s  are independent. Also, we have
\begin{equation}
    \label{eqn:powqq2}
    P(C_{\ell})\leq P(\mbox{no arrival in $\J_{3\ell+1}$}) + 1-p < 1
\end{equation}
for large enough $t$, where $p$ is a uniform lower bound such that $P(\hat{F}_j) \ge p$ for all $j$ provided by Lemma \ref{lem:infinite_many_F}. 

Introduce the atypical events:
\begin{eqnarray}
    B &=& \bigcup_{(i,k): \tau^h_i \in [s,s+t] \mbox{~or~} \tau^h_k + \Delta \in [s,s+t], i < k, \tau^h_k + \Delta - \tau^h_i >  \sqrt{t}} \hat{B}_{ik}, \nonumber \\ \text{ and } \nonumber\\
    \tilde{B} &=& \bigcup_{(i,k):\tau^h_i<s,s+t<\tau^h_k+\Delta} \hat{B}_{ik}. \nonumber
\end{eqnarray}
The events $B$ and $\tilde{B}$ are the events  that an adversary tree catches up with an honest block far ahead. Then we have
\begin{eqnarray}
\label{eqn:powqqq2}P(B_{s,s+t}) &\leq& P(\bigcap_{j: \tau^h_j \in [s,s+t]} U_j^c) + P(B)+P(\tilde B)+P(\bigcap_{\ell=0}^{\sqrt{t}/3} C_{\ell})\nonumber \\
&=& P(\bigcap_{j: \tau^h_j \in [s,s+t]} U_j^c) + P(B)+ P(\tilde B)+(P(C_{\ell}))^{\sqrt{t}/3}\nonumber \\
&\leq&  e^{-c_2 t}+ P(B)+ P(\tilde B) + (P(C_\ell))^{\frac{\sqrt{t}}{3}}
\end{eqnarray}
for some positive constant $c_2$ when $t$ is large, where the equality is due to independence. Next we will bound the atypical events $B$ and $\tilde{B}$.
Consider the following events
\begin{eqnarray*}
D_1&=&\{\#\{i: \tau^h_i\in (s-\sqrt{t}-\Delta,s+t+\sqrt{t}+\Delta)\} >2\lambda_h t\} \label{eqn:powD1}\\
D_2&=&\{ \exists i,k: \tau^h_i  \in (s,s+t), (k-i)<\frac{\sqrt{t}}{2\lambda_h}, \tau^h_k-\tau^h_i+\Delta>\sqrt{t}\} \label{eqn:powD2}\\
D_3&=&\{ \exists i,k: \tau^h_k+\Delta \in (s,s+t), (k-i)<\frac{\sqrt{t}}{2\lambda_h}, \tau^h_k-\tau^h_i+\Delta>\sqrt{t}\} \label{eqn:powD3}
\end{eqnarray*}
In words, $D_1$ is the event of atypically many honest arrivals in $(s-\sqrt{t}-\Delta,s+t+\sqrt{t}+\Delta)$ while $D_2$ and $D_3$ are the events that there exists an interval of length $\sqrt{t}$ with at least one endpoint inside $(s,s+t)$ with atypically small number of arrivals. Since the number of honest arrivals in $(s,s+t)$ is Poisson with parameter $\lambda_h t$, we have from the memoryless property of the Poisson process that  $P(D_1)\leq e^{-c_0t}$ for some constant $c_0=c_0(\lambda_a,\lambda_h)>0$ when $t$ is large.  
On the other hand, using the memoryless property and a union bound, and decreasing $c_0$ if needed, we have that $P(D_2)\leq e^{-c_0 \sqrt{t}}$. Similarly, using time reversal, $P(D_3)\leq e^{-c_0\sqrt{t}}$.
Therefore, again using the memoryless property of the Poisson process,
\begin{eqnarray}
P(B)&\leq & P(D_1\cup D_2\cup D_3)+ P(B\cap D_1^c\cap D_2^c\cap D_3^c)\nonumber\\
&\leq & e^{-c_0 t} + 2e^{-c_0\sqrt{t}}+\sum_{i=1}^{2\lambda_h t} \sum_{k: k-i>\sqrt{t}/2\lambda_h} P(\hat{B}_{ik}) \\
&\leq & e^{-c_3\sqrt{t}},
\label{eqn:powPB}
\end{eqnarray}
for large $t$, where $c_3>0$ are constants that may depend on $\lambda_a,\lambda_h$ and the last inequality is due to Lemma \ref{lem:PBik_pow} . 
We next claim that there exists a constant $\alpha>0$ such that, for all $t$ large,
\begin{equation}
    \label{eqn:powPtB}
    P(\tilde B)\leq e^{- \alpha t}.
\end{equation}
Indeed, we have that
\begin{eqnarray}
&&P(\tilde B) \nonumber\\
&=& \sum_{i<k} \int_0^s P(\tau^h_i\in d\theta) P(\hat{B}_{ik}, \tau^h_k-\tau^h_i+\Delta>s + t-\theta)\nonumber \\
&\leq & \sum_i \int_0^s P(\tau^h_i\in d\theta) \sum_{k:k>i} P(\hat{B}_{ik})^{1/2} P(\tau^h_k-\tau^h_i+\Delta>s + t-\theta)^{1/2}.\nonumber\\
\label{eqn:powPtB1}
\end{eqnarray}
The tails of the Poisson distribution yield the existence of constants $c,c'>0$ so that
\begin{eqnarray}
    \label{eqn:powPoisson_tail}
    &&P(\tau^h_k-\tau^h_i + \Delta>s+t-\theta)\\
    &\leq& \left\{
    \begin{array}{ll}
    1,& (k-i)>c(s+t-\theta-\Delta)\\
    e^{-c'(s+t-\theta-\Delta)},& (k-i)\leq c(s+t-\theta-\Delta).
    \end{array}\right.
\end{eqnarray} 
Lemma \ref{lem:PBik_pow} and \eqref{eqn:powPoisson_tail} yield that there exists a constant $\alpha>0$ so that
\begin{equation}
    \label{eqn:powPtB2}
    \sum_{k: k>i} P(\hat{B}_{i,k})^{1/2}P(\tau^h_k-\tau^h_i>s+t-\theta-\Delta)^{1/2} \leq e^{-2\alpha(s+t-\theta-\Delta)}.
\end{equation}
Substituting this bound in \eqref{eqn:powPtB1} and using that $\sum_i P(\tau^h_i\in d\theta)=d\theta$ gives
\begin{eqnarray}
\label{eqn:powPtB3}
P(\tilde B)&\leq &
\sum_{i} \int_0^s 
P(\tau^h_i\in d\theta) e^{-2\alpha (s+t-\theta-\Delta)}\nonumber\\
&\leq& \int_0^s e^{-2\alpha (s+t-\theta-\Delta)} d\theta
\leq \frac{1}{2\alpha} e^{-2\alpha(t-\Delta)} \leq e^{-\alpha t},
\end{eqnarray}
for $t$ large, proving \eqref{eqn:powPtB}.

Combining (\ref{eqn:powPB}), (\ref{eqn:powPtB3}) and (\ref{eqn:powqqq2}) concludes the proof of step 1.

In step two, we prove for any $\varepsilon > 0$ by recursively applying the bootstrapping procedure in step 1.
Assume the following statement is true: for any $\theta \geq m$ there exist constants $\bar a_\theta,\bar A_\theta$ so that for all $s,t\geq 0$,
\begin{equation}
\label{eqn:powqst_basic}
\tilde{q}[s,s+t] \leq \bar A_\theta \exp(-\bar a_\theta t^{1/\theta}).
\end{equation}
By step 1, it holds for $m = 2$.

Divide $[s,s+t]$ into $t^{\frac{m-1}{2m-1}}$ sub-intervals of length $t^{\frac{m}{2m-1}}$, so that the $r$ th sub-interval is:
$$\J_r : = [s+  (r-1) t^{\frac{m}{2m-1}}, s+ rt^{\frac{m}{2m-1}}].$$

Now look at the first, fourth, seventh, etc sub-intervals, i.e. all the $r = 1 \mod 3$ sub-intervals. Introduce the event that in the $\ell$-th $1 \mod 3$th sub-interval,
an adversary tree that is rooted at a honest block arriving in that sub-interval or in the previous ($0 \mod 3$) sub-interval catches up with a honest block in that sub-interval or in the next ($2 \mod 3$) sub-interval. 
Formally,
$$C_{\ell}=\bigcap_{j: \tau^h_j \in \J_{3\ell+1}}
U_j^c \cup \left(\bigcup_{(i,k): \tau^h_j - t^{\frac{m}{2m-1}} < \tau^h_i < \tau^h_j, \tau^h_j < \tau^h_k +\Delta < \tau^h_j +t^{\frac{m}{2m-1}} } \hat{B}_{ik} \right).$$
Note that for distinct $\ell$, the events $C_\ell$'s  are independent. Also by (\ref{eqn:powqst_basic}), we have
\begin{equation}
    \label{eqn:powpcl}
    P(C_{\ell})\leq  A_m \exp(-\bar a_m t^{1/(2m-1)}).
\end{equation}

Introduce the atypical events:
\begin{eqnarray}
    B &=& \bigcup_{(i,k): \tau^h_i \in [s,s+t] \mbox{~or~} \tau^h_k + \Delta \in [s,s+t], i < k, \tau^h_k + \Delta - \tau^h_i >  t^{\frac{m}{2m-1}}} \hat{B}_{ik}, \nonumber \\ \text{ and } \nonumber\\
    \tilde{B} &=& \bigcup_{(i,k):\tau^h_i<s,s+t<\tau^h_k+\Delta} \hat{B}_{ik}.  \nonumber
\end{eqnarray}
The events $B$ and $\tilde{B}$ are the events  that an adversary tree catches up with an honest block far ahead. Following the calculations in step 1, we have
\begin{eqnarray}
    P(B) &\leq& e^{-c_1 t^{\frac{m}{2m-1}}}\\
    P(\tilde{B}) &\leq& e^{- \alpha t},
\end{eqnarray}
for large $t$, where $c_1$ and $\alpha$ are some positive constant.

Then we have
\begin{eqnarray}
\label{eqn:powqst_induc}
\tilde{q}[s,s+t] &\leq& P(\bigcap_{j: \tau^h_j \in [s,s+t]} U_j^c) + P(B)+P(\tilde B)+P(\bigcap_{\ell=0}^{t^{\frac{m-1}{2m-1}}/3} C_{\ell})\nonumber \\
&=& P(\bigcap_{j: \tau^h_j \in [s,s+t]} U_j^c) + P(B)+ P(\tilde B)+(P(C_{\ell}))^{t^{\frac{m-1}{2m-1}}/3}\nonumber \\
&\leq&  e^{-c_2 t}+ e^{-c t^{\frac{m}{2m-1}}}+ e^{- \alpha t}  \nonumber \\
&& \;+\; ( A_m \exp(-\bar a_m t^{1/(2m-1)}))^{t^{\frac{m-1}{2m-1}}/3} \nonumber \\
&\leq& \bar A'_m \exp(-\bar a'_m t^{\frac{m}{2m-1}})
\end{eqnarray}
for large $t$, where $A'_m$ and $a'_m$ are some positive constant.

So we know the statement in (\ref{eqn:powqst_basic}) holds for all $\theta \geq \frac{2m-1}{m}$. Start with $m_1=2$, we have a recursion equation $m_k = \frac{2m_{k-1}-1}{m_{k-1}}$ and we know (\ref{eqn:powqst_basic}) holds for all $\theta \geq m_k$. It is not hard to see that $m_k = \frac{k+1}{k}$ and thus $\lim_{k\rightarrow\infty} m_k = 1$, which concludes the lemma.

%% file: Appendix_chia.tex
\section{Proofs for Section \ref{sec:chia_analysis}}

Notations used in this section are defined in \S \ref{sec:defin}.

\subsection{The adversary tree via branching random walks}
\label{app:brw}
We first give a description of the (dual of the)
adversary tree in terms of a 
Branching Random Walk (BRW). Such a representation appears already in \cite{pittel94,Drmota}, but we use here 
the standard language from, e.g., \cite{aidekon,shi}.

Consider the collection of $k$ tuples of positive integers, 
$\I_k =\{(i_1,\ldots,i_k)\}$, and  set $\I=\cup_{k>0} 
\I_k$. We consider elements of $\I$ as labelling the vertices of a rooted 
infinite tree, with $\I_k$ labelling the vertices at generation $k$ as follows:
the vertex $v = (i_1,\ldots,i_k)\in \I_k$ 
is the $i_k$-th child of vertex $(i_1,\ldots,i_{k-1})$
at level $k-1$.
An example of labelling is given in Figure~\ref{fig:label_tree}.
For such $v$ we also let $v^j=(i_1,\ldots,i_j)$, 
$j=1,\ldots,k$, denote the 
ancestor of $v$ at level $j$, with $v^k=v$. 
For notation convenience, we set $v^0=0$ as the root of the tree.

\begin{figure}[h]
\centering
\includegraphics[width=0.4\textwidth]{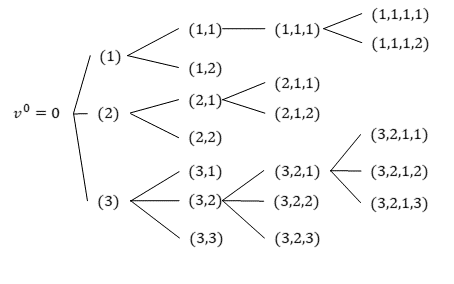}
\caption{Labelling the vertices of a rooted infinite tree.} \label{fig:label_tree}
\end{figure}

Next, let $\{\E_v\}_{v\in \I}$ be an i.i.d. family of exponential 
random variables
of parameter $\lambda_a$. For $v=(i_1,\ldots,i_k)\in \I_k$,
let $\W_v=\sum_{j\leq i_k} \E_{(i_1,\ldots,i_{k-1},j)}$ and
let $S_v=\sum_{j\leq k} \W_{v^j}$. This creates a labelled tree, with the following
interpretation: for $v=(i_1,\ldots,i_j)$,  the $W_{v^j}$ are the waiting
for $v^j$ to appear, measured from the appearance of $v^{j-1}$, and 
$S_v$ is the appearance time of $v$. A moments thought ought to convince the 
reader that the tree $S_v$ is a description of the adversary tree,
sorted by depth. 

Let $S^*_k=\min_{v\in \I_k} S_v$. Note that $S^*_k$ is the time of appearance
of a block at level $k$ and therefore we have
\begin{equation}
  \label{eq-Ofer1}
  \{D_0(t)\leq k\}=\{S^*_k\geq  t\}.
\end{equation}

$S^*_k$ is the minimum of a standard BRW. Introduce, for $\theta<0$,
the moment generating 
function
\begin{eqnarray*}
\Lambda(\theta)&=&\log \sum_{v\in \I_1} E(e^{\theta S_v})=
\log \sum_{j=1}^\infty E (e^{\sum_{i=1}^j \theta \E_i})\\&=& 
\log \sum_{j=1}^\infty (E(e^{\theta \E_1}))^j=
\log \frac{E(e^{\theta \E_1})}{1-E(e^{\theta\E_1})}.
\end{eqnarray*}
Due to the exponential law of $\E_1$,
$E(e^{\theta \E_1})= \frac{\lambda_a}{\lambda_a-\theta}$ and therefore
$\Lambda(\theta)=\log(-\lambda_a/\theta)$. 

An important role is played by $\theta^*=-e\lambda_a$, for which 
$\Lambda(\theta^*)=-1$ and 
$$\sup_{\theta<0} \left(\frac{\Lambda(\theta)}{\theta}\right)= \frac{\Lambda(\theta^*)}{\theta^*}=\frac{1}{\lambda_a e}=\frac{1}{|\theta^*|}.$$
Indeed, 
see 
e.g \cite[Theorem 1.3]{shi}, we have the following.
\begin{lemma}
  \label{prop-1}
$$\lim_{k\to\infty} \frac{S^*_k}{k}= \sup_{\theta<0} \left(\frac{\Lambda(\theta)}{\theta}\right)=\frac{1}{|\theta^*|}, \quad a.s.$$
\end{lemma}
In fact, much more is known, see e.g. \cite{hushi}. 
\begin{lemma}
  \label{prop-2}
  There exist explicit constants $c_1>c_2>0$ 
so that the sequence
$ S^*_k-k/\lambda_a e-c_1\log k$ is tight,
and
$$\liminf_{k\to \infty} S^*_k-k/\lambda_a e-c_2\log k=\infty, a.s.$$
\end{lemma}

Note that Lemmas \ref{prop-1}, \ref{prop-2}
and \eqref{eq-Ofer1} imply in particular
that  $D_0(t)\leq e\lambda_a t$ for all large $t$, a.s., and also that
\begin{equation}
  \label{eq-Ofer2}
\mbox{\rm  if
$e\lambda_a >\lambda_h$ then } D_0(t)> \lambda_ht\; \mbox{\rm for all large $t$, a.s.}.
\end{equation}

With all these preparations, we can give a simple proof for Lemma \ref{theo-tail}.
\begin{proof}
  We use a simple upper bound.
  Note that by \eqref{eq-Ofer1},
  \begin{equation}
    \label{eq-Ofer22}
    P(D_0(t)\geq m)=P(S^*_{ m}\leq t)
  \leq \sum_{v\in \I_{m}} P(S_v\leq t).
\end{equation}
  For $v=(i_1,\ldots,i_k)$, set $|v|=i_1+\cdots+i_k$. Then,
  we have that $S_v$ has the same law as 
  $\sum_{j=1}^{|v|} \E_j$. Thus, by Chebycheff's inequality,
  for $v\in \I_{m}$,
  \begin{equation}
    \label{eq-Ofer3}
    P(S_v\leq  t)\leq Ee^{\theta S_v} e^{-\theta t}=\left(\frac{\lambda_a}{\lambda_a-\theta} \right)^{|v|} e^{-\theta t}.
\end{equation}
  But 
\begin{align}
    \label{eq-Ofer4}
    \sum_{v\in \I_{m}}\left(\frac{\lambda_a}{\lambda_a-\theta}\right)^{|v|}
  &=\sum_{i_1\geq 1,\ldots,i_{m}\geq 1} \left(\frac{\lambda_a}{\lambda_a-\theta}\right)^{\sum_{j=1}^m i_j} \\
  &=\left(\sum_{i\geq 1}\left(\frac{\lambda_a}{\lambda_a-\theta}\right)^i\right)^m= \left(-\frac{\theta}{\lambda_a}\right)^{-{m}}.
\end{align}
Combining \eqref{eq-Ofer3}, \eqref{eq-Ofer4}, we have
\begin{equation*}
    P(D_0(t)\geq m) \leq \left(-\frac{\theta}{\lambda_a}\right)^{-m}e^{-\theta t},
\end{equation*}
and optimizing over $\theta$ we have when $\theta = -m/t$,
\begin{equation*}
    P(D_0(t)\geq m) \leq \left(\frac{e\lambda_a t}{m}\right)^m.
\end{equation*}
\end{proof}

\subsection{Proof of Lemma \ref{lem:infinite_many_F_chia}}
\label{app:proof_delay}

In this proof, let $r_h := \frac{\lambda_h}{1+\lambda_h \Delta}$.

The random processes of interest start from time $0$. To look at the system in stationarity, let us extend them to $-\infty < t < \infty$. More specifically, define $\tau^h_{-1}, \tau^h_{-2}, \ldots$ such that together with $\tau^h_0, \tau^h_1, \ldots$ we have a double-sided infinite Poisson process of rate $\lambda_h$. Also, for each $i < 0$, we define an independent copy of a random adversary tree $ \tt_i$ with the same distribution as $\tt_0$.  And we extend the definition of $\T_h(t)$ and $D_h(t)$ to $t < 0$: the last honest block mined at $\tau^h_{-1} < 0$ and all honest blocks mined within $(\tau^h_{-1}-\Delta,\tau^h_{-1})$ appear in $\T_h(t)$ at their respective mining times to form the level $-1$, and the process repeats for level less than $-1$; let $D_h(t)$ be the level of the last honest arrival before $t$ in $\T_h(t)$, i.e., $D_h(t) = \ell$ if $\tau^h_i \leq t < \tau^h_{i+1}$ and the $i$-th honest block appears at level $\ell$ of $\T_h(t)$.

These extensions allow us to extend the definition of $E_{ij}$ to all $i,j$, $-\infty < i< j < \infty$, and define $E_j$ and $\hat{E}_j$ to be:
$$ E_j = \bigcap_{i < j} E_{ij}$$ and 
$$ \hat{E}_j = E_j \cap U_j.$$

Note that $\hat{E}_j \subset \hat{F}_j$, so to prove that $\hat{F}_j$ has a probability bounded away from $0$ for all $j$, all we need is to prove that $\hat{E}_j$ has a non-zero probability.

Recall that we have defined the event $\hat{B}_{ik}$ in \S \ref{sec:defin} as:
\begin{equation}
\label{eqn:hat_B}
    \hat{B}_{ik} = \mbox{event that $D_i(\sum_{m = i}^{k-1} R_m + \Delta + \tau^h_i) \ge  D_h(\tau^h_{k-1}) - D_h(\tau^h_i+\Delta)$}.
\end{equation}

Following the idea in Lemma~\ref{lem:F_j}, we have 
$$E_j \cap U_j = \bigcap_{i < j} E_{ij} \cap U_j =\left(\bigcap_{i < j <k} \hat{B}_{ik}^c\right) \cap U_j.$$ 
Hence $E_j\cap U_j$ has a time-invariant dependence on $\{\zz_i\}$, which means that $p = P(\hat{E}_j)$ does not depend on $j$. Then we can just focus on $P(\hat{E}_0)$. This is the last step to prove.
\begin{eqnarray*}
P(\hat{E}_0) &=& P(E_0|U_0)P(U_0)  \\
 &=& P(E_0|U_0)P(R_0>\Delta)P(R_{-1}>\Delta) \\
 &=& e^{-2\lambda_h \Delta} P(E_0|U_0).
\end{eqnarray*}

It remains to show that $P(E_0|U_0)>0$. We have
\begin{eqnarray*}
    E_0 &=& \mbox{event that } D_i(\sum_{m = i}^{k-1} R_m + \Delta + \tau^h_i) < D_h(\tau^h_{k-1}) - D_h(\tau^h_i+\Delta)\\
    && \;\;\;\;\;\;\;\;\;\;\;\; \mbox{for all $k > 0$ and $i < 0$},
\end{eqnarray*}
then
\begin{equation}
    E_0^c = \bigcup_{k >0, i < 0}  \hat{B}_{ik}.
\end{equation}

Let us fix a particular $n > 2\lambda_h\Delta > 0$, and define:
\begin{eqnarray*}
    G_n &=& \mbox{event that} D_m(3n/\lambda_h + \tau^h_m) = 0 \\
    && \;\;\;\;\;\;\;\;\;\;\;\; \mbox{for $m = -n, -n+1, \ldots, -1, 0, +1, \ldots, n-1,n$}
\end{eqnarray*}
Then 
\begin{eqnarray}
P(E_0 | U_0) & \ge & P(E_0|U_0,G_n)P(G_n|U_0) \nonumber\\
& = & \left ( 1 - P(\cup_{k>0,i<0} \hat{B}_{ik}|U_0,G_n) \right) P(G_n|U_0)\nonumber\\
& \ge & \left ( 1 - \sum_{k>0,i<0} P(\hat{B}_{ik}|U_0,G_n) \right) P(G_n|U_0) \nonumber\\
& \ge &  ( 1 - a_n - b_n) P(G_n|U_0) \label{eqn:up_bound}
\end{eqnarray}
where
\begin{eqnarray}
a_n & := & \sum_{(i,k): -n \le i < 0 < k \le n} P(\hat{B}_{ik}|U_0,G_n)\\
b_n & := & \sum_{(i,k):  i < -n \text{~or~} k > n }P(\hat{B}_{ik}|U_0,G_n).
\end{eqnarray}

Using (\ref{eq:Ofertheo}), we can bound $P(\hat{B}_{ik}|U_0, G_n)$. Consider two cases:

\noindent
{\bf  Case 1:} $-n \le i <0 <  k \le n$:
\begin{eqnarray*}
P(\hat{B}_{ik}|U_0,G_n) &=& P(\hat{B}_{ik}|U_0, G_n,\sum_{m = i}^{k-1} R_m + \Delta \leq 3n/\lambda_h) \\
&& \;+\;P(\sum_{m = i}^{k-1} R_m + \Delta > 3n/\lambda_h |U_0,G_n) \\
& \leq & P(\sum_{m = i}^{k-1} R_m + \Delta > 3n/\lambda_h |U_0,G_n) \\
& \leq & P(\sum_{m = i}^{k-1} R_m > 5n/(2\lambda_h) |U_0) \\
& \leq & P(\sum_{m = i}^{k-1} R_m  > 5n/(2\lambda_h))/P(U_0)  \\
& \leq & A_1 e^{-\alpha_1 n} 
\end{eqnarray*}
for some positive  constants $A_1, \alpha_1$ independent of $n,k,i$. The last inequality follows from the fact that $R_i$'s are iid 
exponential random variables of  mean $1/\lambda_h$. Summing these terms, we have:
\begin{eqnarray*}
a_n & = & \sum_{(i,k): -n \le i < 0 < k \le n} P(B_{ik}|U_0, G_n) \\
& \leq & \sum_{(i,k): -n \le i < 0 < k \le n} A_1 e^{-\alpha_1 n} : = \bar{a}_n,
\end{eqnarray*}
which is bounded and moreover $\bar{a}_n\rightarrow 0$ as $n \rightarrow \infty$.

\noindent {\bf  Case 2:} $k>n \text{~or~} i<-n$:

For $0<\varepsilon<1$, let us define event $W^{\varepsilon}_{ik}$ to be:
\begin{equation}
\label{eq:honest_growth}
    W^{\varepsilon}_{ik} = \mbox{event that $D_h(\tau^h_{k-1}) - D_h(\tau^h_i+\Delta) \geq (1-\varepsilon)\frac{r_h}{\lambda_h}(k-i-1)$}.
\end{equation}
Then we have
\begin{equation*}
    P(\hat{B}_{ik}|U_0,G_n) \leq P(\hat{B}_{ik}|U_0, G_n,W^{\varepsilon}_{ik}) + P({W^{\varepsilon}_{ik}}^c |U_0,G_n).
\end{equation*}

We first bound $P({W^{\varepsilon}_{ik}}^c |U_0,G_n)$:
\begin{eqnarray}
P({W^{\varepsilon}_{ik}}^c |U_0,G_n) &\leq& P({W^{\varepsilon}_{ik}}^c | \tau^h_{k-1} - \tau^h_i - \Delta> \frac{k-i-1}{(1+\varepsilon)\lambda_h}) \nonumber \\
&& \;+\;P(\tau^h_{k-1} - \tau^h_i - \Delta \leq \frac{k-i-1}{(1+\varepsilon)\lambda_h}) \nonumber \\
& \leq & P({W^{\varepsilon}_{ik}}^c | \tau^h_{k-1} - \tau^h_i -\Delta > \frac{k-i-1}{(1+\varepsilon)\lambda_h}) \nonumber \\
&& \;+\;e^{-\Omega(\varepsilon^2 (k-i-1))} \nonumber \\
& \leq & e^{-\Omega(\varepsilon^4 (k-i-1))} + e^{-\Omega(\varepsilon^2 (k-i-1))}  \nonumber \\
& \leq & A_2 e^{-\alpha_2 (k-i-1)}
\label{eqn:prob_event}
\end{eqnarray}
for some positive  constants $A_2, \alpha_2$ independent of $n,k,i$, where the second inequality follows from the Erlang tail bound and the third inequality follows from Proposition \ref{prop:bound-1} .

Meanwhile, we have 
\begin{eqnarray*}
&~&P(\hat{B}_{ik}|U_0, G_n,W^{\varepsilon}_{ik})  \\
&\leq& P(D_i(\sum_{m = i}^{k-1}R_m + \Delta + \tau^h_i) \geq (1-\varepsilon)\frac{r_h}{\lambda_h}(k-i-1) | U_0,G_n,W^{\varepsilon}_{ik})  \\
& \leq & P(D_i(\sum_{m = i}^{k-1}R_m + \Delta + \tau^h_i) \geq (1-\varepsilon)\frac{r_h}{\lambda_h}(k-i-1) \\
&&\;\;\;\;\;\;| U_0,G_n,W^{\varepsilon}_{ik},\sum_{m = i}^{k-1}R_m + \Delta \leq (k-i-1)\frac{r_h+\lambda_a e}{2\lambda_a e}\frac{1}{\lambda_h}) \\
&& \;+\;P(\sum_{m = i}^{k-1}R_m + \Delta  > (k-i-1)\frac{r_h+\lambda_a e}{2\lambda_a e}\frac{1}{\lambda_h}|U_0,G_n,W^{\varepsilon}_{ik}) \\
&\leq&  P(\sum_{m = i}^{k-1}R_m + \Delta > (k-i-1)\frac{r_h+\lambda_a e}{2\lambda_a e}\frac{1}{\lambda_h}|U_0,G_n,W^{\varepsilon}_{ik})\\
&& \;+\; \left( \frac{r_h+\lambda_a e}{2(1-\varepsilon)r_h} \right)^{(1-\varepsilon)\frac{r_h}{\lambda_h}(k-i-1)}
\end{eqnarray*}
where the first term in the last inequality follows from (\ref{eq:Ofertheo}), and the second term can also be bounded:
\begin{eqnarray*}
&~&P(\sum_{m = i}^{k-1}R_m + \Delta > (k-i-1)\frac{r_h+\lambda_a e}{2\lambda_a e}\frac{1}{\lambda_h}|U_0,G_n,W^{\varepsilon}_{ik})\\
&=& P(\sum_{m = i}^{k-1}R_m + \Delta > (k-i-1)\frac{r_h+\lambda_a e}{2\lambda_a e}\frac{1}{\lambda_h}|U_0, W^{\varepsilon}_{ik}) \\
&\leq& P(\sum_{m = i}^{k-1}R_m + \Delta > (k-i-1)\frac{r_h+\lambda_a e}{2\lambda_a e}\frac{1}{\lambda_h})/P(U_0, W^{\varepsilon}_{ik}) \\
&\leq& A_3 e^{-\alpha_3(k-i-1)}
\end{eqnarray*}
for some positive  constants $A_3, \alpha_3$ independent of $n,k,i$. The last inequality follows from the fact that $(r_h+\lambda_a e)/(2\lambda_a e) > 1$ and the $R_i$'s have mean $1/\lambda_h$, while $P(U_0, W^{\varepsilon}_{ik})$ is a event with high probability as we showed in (\ref{eqn:prob_event}).

Then we have 
\begin{eqnarray}
\label{eqn:B_hat_bound}
     P(\hat{B}_{ik}|U_0,G_n) &\leq& A_2 e^{-\alpha_2 (k-i-1)} + \left( \frac{r_h+\lambda_a e}{2(1-\varepsilon)r_h} \right)^{(1-\varepsilon)\frac{r_h}{\lambda_h}(k-i-1)} \nonumber \\
     &+& A_3 e^{-\alpha_3(k-i-1)}.
\end{eqnarray}
Summing these terms, we have:
\begin{eqnarray*}
b_n & = & \sum_{(i,k):  i<-n \text{~or~} k > n}P(\hat{B}_{ik}|U_0,G_n) \\
& \le  &  \sum_{ (i,k): i<-n \text{~or~} k > n}  [A_2 e^{-\alpha_2 (k-i-1)} \\
&& \;\;+\; \left( \frac{r_h+\lambda_a e}{2(1-\varepsilon)r_h} \right)^{(1-\varepsilon)\frac{r_h}{\lambda_h}(k-i-1)} + A_3 e^{-\alpha_3(k-i-1)}] \\
&:=& \bar{b}_n
\end{eqnarray*}
which is bounded and moreover $\bar{b}_n \rightarrow 0$ as $n \rightarrow \infty$ when we set $\varepsilon$ to be small enough such that $\frac{r_h+\lambda_a e}{2(1-\varepsilon)r_h}<1$.

Substituting these bounds in (\ref{eqn:up_bound}) we finally get:
\begin{equation}
    P(E_0|U_0) > [1- (\bar{a}_n+ \bar{b}_n)]P(G_n|U_0)
\end{equation}
By setting $n$ sufficiently large such that $\bar{a}_n$ and $\bar{b}_n$ are sufficiently small, we conclude that $P(\hat{E}_0)> 0$.

\subsection{Proof of Lemma \ref{lem:time_strong_chia}}
\label{app:proof_time_strong}
We divide the proof in to two steps. In the first step, we prove for $\varepsilon = 1/2$.

Recall that we have defined event $\hat{B}_{ik}$ in \S \ref{sec:defin} as:
\begin{equation*}
    \hat{B}_{ik} = \mbox{event that $D_i(\sum_{m = i}^{k-1} R_m + \Delta + \tau^h_i) \ge  D_h(\tau^h_{k-1}) - D_h(\tau^h_i+\Delta)$}.
\end{equation*}

Note that from Lemma~\ref{lem:infinite_many_F_chia} and similar to inequality (\ref{eqn:B_hat_bound}), we have 
\begin{equation}
\label{eqn:PBik}
  P(\hat{B}_{ik}) \leq e^{-c_1 (k-i-1)}
\end{equation}
for some positive constants $c_1$.

And by Lemma \ref{lem:F_j}, we have
\begin{equation}
    \hat{F}_j^c = F_j^c \cup U_j^c = \left(\bigcup_{(i,k): i < j <k} \hat{B}_{ik}\right) \cup U_j^c.
\end{equation}
Divide $[s,s+t]$ into $\sqrt{t}$ sub-intervals of length $\sqrt{t}$, so that the $r$ th sub-interval is:
$$\J_r : = [s+  (r-1) \sqrt{t}, s+ r\sqrt{t}].$$

Now look at the first, fourth, seventh, etc sub-intervals, i.e. all the $r = 1 \mod 3$ sub-intervals. Introduce the event that in the $\ell$-th $1 \mod 3$th sub-interval,
an adversary tree that is rooted at a honest block arriving in that sub-interval or in the previous ($0 \mod 3$) sub-interval catches up with a honest block in that sub-interval or in the next ($2 \mod 3$) sub-interval. 
Formally,
$$C_{\ell}=\bigcap_{j: \tau^h_j \in \J_{3\ell+1}}
U_j^c \cup \left(\bigcup_{(i,k): \tau^h_j - \sqrt{t} < \tau^h_i < \tau^h_j, \tau^h_j < \tau^h_k +\Delta < \tau^h_j +\sqrt{t} } \hat{B}_{ik} \right).$$
Note that for distinct $\ell$, the events $C_\ell$'s  are independent. Also, we have
\begin{equation}
    \label{eqn:qq2}
    P(C_{\ell})\leq P(\mbox{no arrival in $\J_{3\ell+1}$}) + 1-p < 1
\end{equation}
for large enough $t$, where $p$ is a uniform lower bound such that $P(\hat{F}_j) \ge p$ for all $j$ provided by Lemma \ref{lem:infinite_many_F_chia}. 

Introduce the atypical events:
\begin{eqnarray}
    B &=& \bigcup_{(i,k): \tau^h_i \in [s,s+t] \mbox{~or~} \tau^h_k + \Delta \in [s,s+t], i < k, \tau^h_k + \Delta - \tau^h_i >  \sqrt{t}} \hat{B}_{ik}, \nonumber \\ \text{ and } \nonumber\\
    \tilde{B} &=& \bigcup_{(i,k):\tau^h_i<s,s+t<\tau^h_k+\Delta} \hat{B}_{ik}. \nonumber
\end{eqnarray}
The events $B$ and $\tilde{B}$ are the events  that an adversary tree catches up with an honest block far ahead. Then we have
\begin{eqnarray}
\label{eqn:qqq2}P(B_{s,s+t}) &\leq& P(\bigcap_{j: \tau^h_j \in [s,s+t]} U_j^c) + P(B)+P(\tilde B)+P(\bigcap_{\ell=0}^{\sqrt{t}/3} C_{\ell})\nonumber \\
&=& P(\bigcap_{j: \tau^h_j \in [s,s+t]} U_j^c) + P(B)+ P(\tilde B)+(P(C_{\ell}))^{\sqrt{t}/3}\nonumber \\
&\leq&  e^{-c_2 t}+ P(B)+ P(\tilde B) + (P(C_\ell))^{\frac{\sqrt{t}}{3}}
\end{eqnarray}
for some positive constant $c_2$ when $t$ is large, where the equality is due to independence. Next we will bound the atypical events $B$ and $\tilde{B}$.
Consider the following events
\begin{eqnarray*}
D_1&=&\{\#\{i: \tau^h_i\in (s-\sqrt{t}-\Delta,s+t+\sqrt{t}+\Delta)\} >2\lambda_h t\} \label{eqn:D1}\\
D_2&=&\{ \exists i,k: \tau^h_i  \in (s,s+t), (k-i)<\frac{\sqrt{t}}{2\lambda_h}, \tau^h_k-\tau^h_i+\Delta>\sqrt{t}\} \label{eqn:D2}\\
D_3&=&\{ \exists i,k: \tau^h_k+\Delta \in (s,s+t), (k-i)<\frac{\sqrt{t}}{2\lambda_h}, \tau^h_k-\tau^h_i+\Delta>\sqrt{t}\} \label{eqn:D3}
\end{eqnarray*}
In words, $D_1$ is the event of atypically many honest arrivals in $(s-\sqrt{t}-\Delta,s+t+\sqrt{t}+\Delta)$ while $D_2$ and $D_3$ are the events that there exists an interval of length $\sqrt{t}$ with at least one endpoint inside $(s,s+t)$ with atypically small number of arrivals. Since the number of honest arrivals in $(s,s+t)$ is Poisson with parameter $\lambda_h t$, we have from the memoryless property of the Poisson process that  $P(D_1)\leq e^{-c_0t}$ for some constant $c_0=c_0(\lambda_a,\lambda_h)>0$ when $t$ is large.  
On the other hand, using the memoryless property and a union bound, and decreasing $c_0$ if needed, we have that $P(D_2)\leq e^{-c_0 \sqrt{t}}$. Similarly, using time reversal, $P(D_3)\leq e^{-c_0\sqrt{t}}$.
Therefore, again using the memoryless property of the Poisson process,
\begin{eqnarray}
P(B)&\leq & P(D_1\cup D_2\cup D_3)+ P(B\cap D_1^c\cap D_2^c\cap D_3^c)\nonumber\\
&\leq & e^{-c_0 t} + 2e^{-c_0\sqrt{t}}+\sum_{i=1}^{2\lambda_h t} \sum_{k: k-i>\sqrt{t}/2\lambda_h} P(\hat{B}_{ik}) \\
&\leq & e^{-c_3\sqrt{t}},
\label{eqn:PB}
\end{eqnarray}
for large $t$, where $c_3>0$ are constants that may depend on $\lambda_a,\lambda_h$ and the last inequality is due to (\ref{eqn:PBik}). 
We next claim that there exists a constant $\alpha>0$ such that, for all $t$ large,
\begin{equation}
    \label{eqn:PtB}
    P(\tilde B)\leq e^{- \alpha t}.
\end{equation}
Indeed, we have that
\begin{eqnarray}
&&P(\tilde B) \nonumber\\
&=& \sum_{i<k} \int_0^s P(\tau^h_i\in d\theta) P(\hat{B}_{ik}, \tau^h_k-\tau^h_i+\Delta>s + t-\theta)\nonumber \\
&\leq & \sum_i \int_0^s P(\tau^h_i\in d\theta) \sum_{k:k>i} P(\hat{B}_{ik})^{1/2} P(\tau^h_k-\tau^h_i+\Delta>s + t-\theta)^{1/2}.\nonumber\\
\label{eqn:PtB1}
\end{eqnarray}
The tails of the Poisson distribution yield the existence of constants $c,c'>0$ so that
\begin{eqnarray}
    \label{eqn:Poisson_tail}
    &&P(\tau^h_k-\tau^h_i + \Delta>s+t-\theta)\\
    &\leq& \left\{
    \begin{array}{ll}
    1,& (k-i)>c(s+t-\theta-\Delta)\\
    e^{-c'(s+t-\theta-\Delta)},& (k-i)\leq c(s+t-\theta-\Delta).
    \end{array}\right.
\end{eqnarray} 
(\ref{eqn:PBik}) and \eqref{eqn:Poisson_tail} yield that there exists a constant $\alpha>0$ so that
\begin{equation}
    \label{eqn:PtB2}
    \sum_{k: k>i} P(\hat{B}_{i,k})^{1/2}P(\tau^h_k-\tau^h_i>s+t-\theta-\Delta)^{1/2} \leq e^{-2\alpha(s+t-\theta-\Delta)}.
\end{equation}
Substituting this bound in \eqref{eqn:PtB1} and using that $\sum_i P(\tau^h_i\in d\theta)=d\theta$ gives
\begin{eqnarray}
\label{eqn:PtB3}
P(\tilde B)&\leq &
\sum_{i} \int_0^s 
P(\tau^h_i\in d\theta) e^{-2\alpha (s+t-\theta-\Delta)}\nonumber\\
&\leq& \int_0^s e^{-2\alpha (s+t-\theta-\Delta)} d\theta
\leq \frac{1}{2\alpha} e^{-2\alpha(t-\Delta)} \leq e^{-\alpha t},
\end{eqnarray}
for $t$ large, proving \eqref{eqn:PtB}.

Combining (\ref{eqn:PB}), (\ref{eqn:PtB3}) and (\ref{eqn:qqq2}) concludes the proof of step 1.


In step two, we prove for any $\varepsilon > 0$ by recursively applying the bootstrapping procedure in step 1.
Assume the following statement is true: for any $\theta \geq m$ there exist constants $\bar a_\theta,\bar A_\theta$ so that for all $s,t\geq 0$,
\begin{equation}
\label{eqn:qst_basic}
\tilde{q}[s,s+t] \leq \bar A_\theta \exp(-\bar a_\theta t^{1/\theta}).
\end{equation}
By step 1, it holds for $m = 2$.

Divide $[s,s+t]$ into $t^{\frac{m-1}{2m-1}}$ sub-intervals of length $t^{\frac{m}{2m-1}}$, so that the $r$ th sub-interval is:
$$\J_r : = [s+  (r-1) t^{\frac{m}{2m-1}}, s+ rt^{\frac{m}{2m-1}}].$$

Now look at the first, fourth, seventh, etc sub-intervals, i.e. all the $r = 1 \mod 3$ sub-intervals. Introduce the event that in the $\ell$-th $1 \mod 3$th sub-interval,
an adversary tree that is rooted at a honest block arriving in that sub-interval or in the previous ($0 \mod 3$) sub-interval catches up with a honest block in that sub-interval or in the next ($2 \mod 3$) sub-interval. 
Formally,
$$C_{\ell}=\bigcap_{j: \tau^h_j \in \J_{3\ell+1}}
U_j^c \cup \left(\bigcup_{(i,k): \tau^h_j - t^{\frac{m}{2m-1}} < \tau^h_i < \tau^h_j, \tau^h_j < \tau^h_k +\Delta < \tau^h_j +t^{\frac{m}{2m-1}} } \hat{B}_{ik} \right).$$
Note that for distinct $\ell$, the events $C_\ell$'s  are independent. Also by (\ref{eqn:qst_basic}), we have
\begin{equation}
    \label{eqn:pcl}
    P(C_{\ell})\leq  A_m \exp(-\bar a_m t^{1/(2m-1)}).
\end{equation}

Introduce the atypical events:
\begin{eqnarray}
    B &=& \bigcup_{(i,k): \tau^h_i \in [s,s+t] \mbox{~or~} \tau^h_k + \Delta \in [s,s+t], i < k, \tau^h_k + \Delta - \tau^h_i >  t^{\frac{m}{2m-1}}} \hat{B}_{ik}, \nonumber \\ \text{ and } \nonumber\\
    \tilde{B} &=& \bigcup_{(i,k):\tau^h_i<s,s+t<\tau^h_k+\Delta} \hat{B}_{ik}.  \nonumber
\end{eqnarray}
The events $B$ and $\tilde{B}$ are the events  that an adversary tree catches up with an honest block far ahead. Following the calculations in step 1, we have
\begin{eqnarray}
    P(B) &\leq& e^{-c_1 t^{\frac{m}{2m-1}}}\\
    P(\tilde{B}) &\leq& e^{- \alpha t},
\end{eqnarray}
for large $t$, where $c_1$ and $\alpha$ are some positive constant.

Then we have
\begin{eqnarray}
\label{eqn:qst_induc}
\tilde{q}[s,s+t] &\leq& P(\bigcap_{j: \tau^h_j \in [s,s+t]} U_j^c) + P(B)+P(\tilde B)+P(\bigcap_{\ell=0}^{t^{\frac{m-1}{2m-1}}/3} C_{\ell})\nonumber \\
&=& P(\bigcap_{j: \tau^h_j \in [s,s+t]} U_j^c) + P(B)+ P(\tilde B)+(P(C_{\ell}))^{t^{\frac{m-1}{2m-1}}/3}\nonumber \\
&\leq&  e^{-c_2 t}+ e^{-c t^{\frac{m}{2m-1}}}+ e^{- \alpha t}  \nonumber \\
&& \;+\; ( A_m \exp(-\bar a_m t^{1/(2m-1)}))^{t^{\frac{m-1}{2m-1}}/3} \nonumber \\
&\leq& \bar A'_m \exp(-\bar a'_m t^{\frac{m}{2m-1}})
\end{eqnarray}
for large $t$, where $A'_m$ and $a'_m$ are some positive constant.

So we know the statement in (\ref{eqn:qst_basic}) holds for all $\theta \geq \frac{2m-1}{m}$. Start with $m_1=2$, we have a recursion equation $m_k = \frac{2m_{k-1}-1}{m_{k-1}}$ and we know (\ref{eqn:qst_basic}) holds for all $\theta \geq m_k$. It is not hard to see that $m_k = \frac{k+1}{k}$ and thus $\lim_{k\rightarrow\infty} m_k = 1$, which concludes the lemma.

%% file: Appendix_per_live.tex
\section{Proof of Persistence and Liveness}
\label{sec:persist_live}

In this section, we will prove Lemma \ref{lem:nak_secure}. Our goal is to generate a transaction ledger that satisfies persistence and liveness as defined in section \ref{sec:results}. Together, persistence and liveness guarantees {\em robust transaction ledger} \cite{backbone}; honest transactions will be adopted to the ledger and be immutable.

\begin{proof}
We first prove persistence by contradiction.
For a chain $\C_t$ with the last block mined at time $t$, let $\C_t^{\lceil\sigma}$ be the chain resulting from pruning a chain $\C_t$ up to $\sigma$, by removing the last blocks at the end of the chain that were mined after time $t-\sigma$. Note that $\C^{\lceil \sigma}$ is a prefix of $\C$, which we denote by $\C^{\lceil \sigma}\preceq \C$. 

Let $\C_t$ denote the longest chain adopted by an honest node with the last block mined at time $t$.  
Suppose there exists a longest chain $\C_t'$ adopted by some honest node with the last block mined at time $t'>t$ and $\C_t^{\lceil \sigma} \not\preceq \C_{t'}$.  
There are a number of honest blocks mined in the time interval $[t-\sigma,t]$, each of which can be in $\C_t$, $\C_{t'}$, or neither. 
We partition the set of honest blocks  generated in that interval with three sets:
$\{\H_t \triangleq\{H_j \in \C_t:\tau_j\in[t-\sigma,t]\}, \H_{t'}\triangleq\{H_j \in \C_{t'}:\tau_j\in[t-\sigma,t]\}$, and $\H_{\rm rest}\triangleq\{H_j \notin \C_t \cup \C_{t'} :\tau_j\in[t-\sigma,t]\}$, depending on which chain they belong to. 

Then we claim that $\C_t^{\lceil \sigma} \not\preceq \C_{t'}$ implies that $\hat{F}_j^c$ holds for all $j$ such that $\tau_j\in[t-\sigma,t]$. 
This in turn implies that $P(\C_t^{\lceil \sigma} \not\preceq \C_{t'}) \leq P(\cap_{j:\tau_j\in[t-\sigma,t]} \hat{F}_j^c)$. 
However, we know that the probability of this happening is as low as $q_\sigma$. 
This follows from the following facts. 
$(i)$ the honest blocks in $\C_t$ does not make it to the longest chain at time $t'$: $H_j \notin \C_{t'}$ for all  $H_j\in\H_t$, which follows from $\C_t^{\lceil \sigma} \not\preceq \C_{t'}$.
$(ii)$ the honest blocks in $\C_{t'}$ does not make it to the longest chain $\C_t$ at time $t$:  
$H_j \notin \C_{t}$ for all  $H_j\in\H_{t'}$, which also follows from $\C_t^{\lceil \sigma} \not\preceq \C_{t'}$.
$(iii)$ the rest of the honest blocks did not make it to either of the above: $H_j \notin \C_t \cup \C_{t'} $  for all $H_j\in \H_{\rm rest}$.

We next prove liveness.
Assume a transaction {\sf tx} is received by all honest nodes at time $t$, then we know that with probability at least $1-q_\sigma$, there exists one honest block $b_j$ mined at time $\tau^h_j$ with $\tau^h_j \in [t,t+\sigma]$ and event $\hat{F}_j$ occurs, i.e., the block $b_j$ and its ancestor blocks will be contained in any future longest chain. Therefore, {\sf tx} must be contained in block $b_j$ or one ancestor block of $b_j$ since {\sf tx} is seen by all honest nodes at time $t < \tau_j$. In either way, {\sf tx} is stabilized forever. Thus, liveness holds.
\end{proof}

%% file: Appendix_Worst_Attack.tex
\section{Proofs for Section \ref{sec:worst_attack}}
\label{sec:lem:worst-attack}



\subsection{Proof of Theorem \ref{lem:worst-attack}}
\label{sec:lem:worst-attack-zero}

Before presenting the full proof for theorem  \ref{lem:worst-attack}, which covers the $\Delta=0$ case, we first describe $\pi_{SZ}$, Sompolinsky and Zohar's strategy of private attack with pre-mining, focusing on some block $b$. We know that if $b=h_j$ is an honest block with index $j$, it will be mined at the tip of the public longest chain $\C$ when $\Delta=0$. In this case, $\pi_{SZ}$ consists of two phases:

\begin{itemize}
    \item \textbf{Pre-mining phase:} Starting from the genesis block, the attacker starts mining blocks in private to build a private chain. When the first honest block $h_1$
    is mined on the genesis block, the attacker does one of two things: i) If the private chain is longer than the public chain at that moment, then the adversary continues mining on the private chain; ii) if the private chain is shorter than the public chain, the attacker abandons the private chain it has been mining on and starts a new private chain on $h_1$ instead. The attacker repeats
    this process with all honest blocks $h_2$, $h_3$, . . . $h_{j-1}$.
    
    \item \textbf{Private attack phase:} After block $h_{j-1}$ is mined, the attacker starts Nakamoto’s private attack from the current private chain it is working on, whether it is off $h_{j-1}$ or the one it has been working on before $h_{j-1}$ depending on which is longer.
\end{itemize}

Note that it is possible for the adversary to attack one of its own blocks. In this case, $b$ is placed at the tip of $\C$, and, kept private until an honest block $h_{j-1}$ is mined at the same depth as $b$. Then, the adversary denotes the chain including $b$ as the longest chain for all honest miners. Hence, we can treat $b$ as if it is an honest block with index $j$, and, the strategy proceeds as described above for all other adversary blocks.  

Having presented an algorithmic description for $\pi_{SZ}$ above, we now identify certain features of $\pi_{SZ}$, which will be used in the proof of theorem \ref{lem:worst-attack}:  
\begin{enumerate}
    \item All of the adversary blocks (except $b$ when it is an adversary block) mined after the genesis block are placed at distinct depths in increasing order of their arrival times.
    \item If an adversary block (except $b$ when it is an adversary block) arrives after an honest block $h_i$ for $i<j$, it is placed at a depth larger than the depth of $h_i$.
    \item None of the paths from adversary blocks to the genesis includes block $b$.
    \item No adversary block (except $b$ when it is an adversary block) is revealed until the attack is successful.
\end{enumerate}

We now proceed with the proof:

\begin{proof}

We first prove part (i) of the theorem, namely the fact that $\pi_{SZ}$ is the worst-attack for preventing persistence with parameter $k$. Consider a sequence of mining times for the honest and adversary blocks such that the persistence of $b$ with parameter $k$ is violated by an adversary following some arbitrary attack strategy $\pi$. Let $\tau_b$ be the mining time of block $b$. Define $t > \tau_b$ as the first time block $b$ disappears from the public longest chain $\C$ after it becomes $k$ deep within $\C$ at some previous time. We will prove this part of the theorem by showing that $\pi_{SZ}$ also succeeds in removing $b$ from $\C$ after it becomes $k$ deep, for the same sequence of block mining times.

Let $\T$ be the blocktree built under $\pi$, and, observe that the public longest chain, $\C(t)$, contains block $b$ at time $t$. By our assumption, we know that at time $t$, there exists a parallel chain $\C'$ with depth greater than or equal to $L(t)$, depth of $\C$ at time $t$, and, $\C'$ does not include $b$. Hence, it also does not include any of the blocks that came to $\C$ after $b$. See Figure \ref{fig:lemma5_1} for a visual example of the chains $\C$ and $\C'$. Let $h_i$ be the last honest block in $\C'$ that is also on $\C$. Such a block $h_i$ must exist; otherwise, these chains could not have grown from the same genesis block. Then, $h_i$ has depth smaller than the depth of $b$. In this context, let $d^h_i$ and $d_b$, $d^h_i<d_b$, denote the depths of $h_i$ and $b$ respectively. Define $H$ as the number of honest blocks mined in the time interval $(\tau^h_i,t]$, and, observe that all of these honest blocks lay in the depth interval $(d^h_i,L(t)]$ of the blocktree $\T(t)$ as there cannot be honest blocks at depths larger than $L(t)$ at time $t$.

Next, consider the portion of $\T(t)$ deeper than $d^h_i$. Let $d:=L(t)-d^h_i$, and, define $A$ as the number of adversary blocks mined in the time interval $(\tau^h_i,t]$. Note that since $\C$ and $\C'$ both include $h_i$, the adversary blocks that are within these chains and have depths greater than $d^h_i$, should have been mined after time $\tau^h_i$. Now, as there can be at most one honest block at every depth due to $\Delta=0$; $H \leq d$. Moreover, at each depth after $d^h_i$, either $\C$ and $\C'$ have two distinct blocks, or, they share the same block, which by definition is an adversary block. Hence, the number of the adversary blocks that are within these chains and have depths greater than $d^h_i$ is at least $d$, which implies $A \geq d$. Hence,
$$A \geq d \geq H.$$
Finally, we know from the definition of persistence that block $b$ has been at least $k$ deep in $\C$ before time $t$, and, there are $d_b-d^h_i$ blocks of distinct depths from $h_i$ to $b$. Consequently, $d \geq (k-1)+(d_b-d^h_i)$. Figure \ref{fig:lemma5_1} displays the interplay between these parameters in the context of an example attack. 

\begin{figure}
    \centering
    \includegraphics[width=\linewidth]{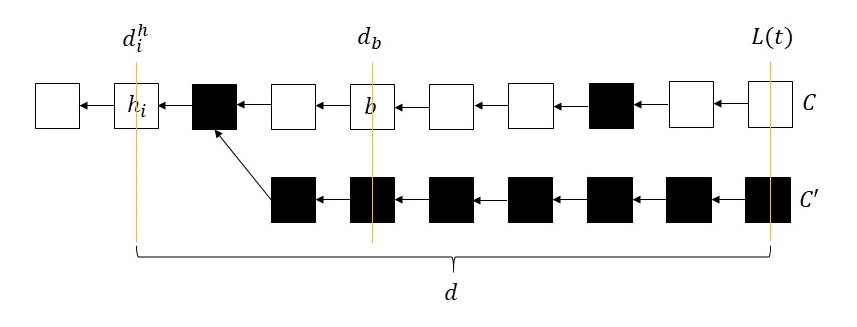}
    \caption{Chains $\C$ and $\C'$ for an arbitrary attack $\pi$. In this example, $k=6$, $H=6$, and, $A=9$. $d_b-d^h_i=3$, and, the attack succeeds at time $t$, at which $b$ is exactly $6$ blocks deep in the chain $\C$. Hence, in this example, $d$ is exactly equal to $(k-1)+(d_b-d^h_i)=5+3=8$}
    \label{fig:lemma5_1}
\end{figure}

\begin{figure}
    \centering
    \includegraphics[width=\linewidth]{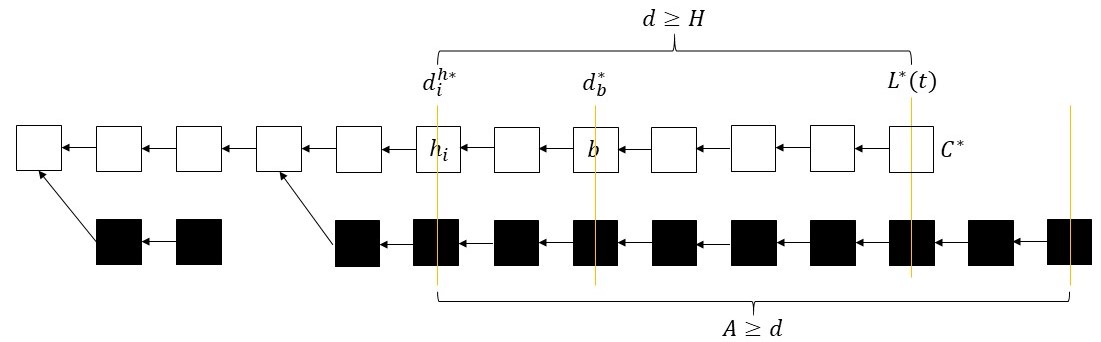}
    \caption{$\C^*$ and the private adversary chain under $\pi_{SZ}$ imposed on the same mining times as in Figure \ref{fig:lemma5_1}. Again, $k=6$, $H=6$, and, $A=9$. Adversary has a private chain at depth $A+d^{h*}_i = 9+d^{h*}_i > (k-1)+d^{*}_b = 5+d^{*}_b$ at time $t$. Note that at time $t$, $b$ is not $k=6$ blocks deep yet. However, the attack will succeed after $b$ is $6$ blocks deep in the chain $\C^*$ since the adversary already has a chain that is at depth greater than $d^{*}_b+(k-1)=d^{*}_b+5$.}
    \label{fig:lemma5_2}
\end{figure}

We now consider an adversary that follows strategy $\pi_{SZ}$. Again, let $\T^*$ be the blocktree built under $\pi_{SZ}$, and, define $d^{h*}_i$ and $d^{*}_b$ as the depths of the blocks $h_i$ and $b$ within $\T^*$. Let $\C^*$ denote the public longest chain under strategy $\pi_{SZ}$. See Figure \ref{fig:lemma5_2} for a visual example of the $\pi_{SZ}$ attack.
We next make the following observations using the properties of $\pi_{SZ}$: Via property (2) of $\pi_{SZ}$, every adversary block mined after time $\tau^h_i$ is placed at a depth higher than $d^{h*}_i$. Via property (1), every one of these adversary blocks mined after time $\tau^h_i$ is placed at a distinct depth. Hence, at time $t$, the deepest adversary block has depth at least $d^{h*}_i+A$. Via property (3), the path from this deepest adversary block to the genesis does not include $b$. Consequently, at time $t$, the adversary following $\pi_{SZ}$, has a private chain that does not include $b$ and is at depth at least $d^{h*}_i+A$. 

Finally, we observe via property (4) of $\pi_{SZ}$ that $\C^*$ contains no adversary blocks (except $b$ when it is an adversary block). Then, at time $t>\tau_b$, $\C^*$ contains $b$, and, it is exactly at depth $L^*(t)=d^{h*}_i+H$ as $\Delta=0$. Finally, to prove that the adversary succeeds under $\pi_{SZ}$, we consider the following two cases:
\begin{itemize}
    \item $b$ is at least $k$-deep in $\C^*$ at time $t$, i.e $L^*(t) \geq (k-1)+d^*_b$. However, since the adversary has a private chain that does not include $b$ and has depth at least $d^{h*}_i+A \geq d^{h*}_i+H = L^*(t)$, the attack is successful.
    \item $b$ is not $k$-deep yet, i.e $L^*(t) < (k-1)+d^{*}_b$. (Figure \ref{fig:lemma5_2} corresponds to this case.) However, the adversary has a private chain that does not include $b$ and is at depth at least, $$d^{h*}_i+A \geq d^{h*}_i+d \geq (k-1)+ d_b + d^{h*}_i - d^h_i.$$
    Moreover, as $\C^*$ does not contain any adversary blocks under $\pi_{SZ}$ (except $b$), $d_b-d^{h}_i \geq d^{*}_b-d^{h*}_i$. Hence, $d^{h*}_i+A \geq (k-1) + d^{*}_b$, implying that the adversary would eventually succeed once $b$ becomes $k$-deep in $\C^*$.
\end{itemize}

This concludes the proof of part (i) of the theorem.
\newline

Second, we prove part (ii) of the theorem, namely the fact that $\pi_{SZ}$ is the worst-attack for preventing liveness with parameter $k$. Consider a sequence of mining times for the honest and adversary blocks such that the liveness of the $k$ consecutive honest blocks starting with $b$ is violated by an adversary following some arbitrary attack strategy $\pi$. Since $b$ is an honest block by assumption, let $b=h_j$ without loss of generality. For each of the $k$ consecutive honest blocks $h_m$, $m=j,..,j+k-1$; define $t_m \geq \tau^h_m$ as the first time block $h_m$ disappeared from the public longest chain $\C$. Let $t^*$ denote the maximum of $t_m$, $m=j,..,j+k-1$. We will prove this part of the theorem by showing that $\pi_{SZ}$ also succeeds in removing each $h_m$, $m=j,..,j+k-1$ from $\C$ by time $t^*$, for the same sequence of block mining times.

Let $\T$ be the blocktree built under $\pi$, and, observe that at time $t_m$, (i) $\C(t_m)$ contains the block $h_m$, (ii) there exists a parallel chain $\C_m$ with depth greater than or equal to $L(t_m)$, depth of $\C$ at time $t_m$, and, $\C_m$ does not include $h_m$. See Figure \ref{fig:lemma5_3} for a visual example of the attack $\pi$. Let $e(m)$ be the index of the last honest block in $\C_m$ that is also on $\C$. Such a block must exist for each $m$; otherwise, the chains $\C_m$ could not have grown from the same genesis block. Let $d^*$ denote the minimum depth of the honest blocks $h_{e(m)}$:
$$d^* = \min_{m=j,..,j+k-1} (d^h_{e(m)}) < d^h_j$$
Let $e^*$ denote the index of the honest block at depth $d^*$. Define $H$ as the number of honest blocks mined in the time interval $(\tau^h_{e^*},t^*]$, and, observe that all of these honest blocks lay in the depth interval $(d^*,L(t^*)]$ as there cannot be honest blocks at depths larger than $L(t^*)$ at time $t^*$.

Next, consider the portion of $\T(t^*)$ deeper than $d^*$. Let $d:=L(t^*)-d^*$, and, define $A$ as the number of adversary blocks mined in the time interval $(\tau^h_{e^*},t^*]$. Note that all of the adversary blocks within the chains $\C_m(t^*)$, $m=j,..,j+k-1$, and $\C(t^*)$ at time $t^*$ that lay in the depth interval $(d^*,L(t^*)]$, should have been mined after time $\tau^h_{e^*}$. Hence, these adversary blocks constitute a subset of the adversary blocks mined in the time interval $(\tau^h_{e^*},t^*]$. As there can be at most one honest block at every depth as $\Delta=0$, $d \geq H$. Moreover, at each depth after $d^*=d^h_{e^*}$, for any given $m$, either $\C$ and $\C_m$ have two distinct blocks, or, they share the same block, which by definition is an adversary block. Hence, the number of the adversary blocks within the chains $\C_m(t^*)$, $m=j,..,j+k-1$, and $\C(t^*)$ at time $t^*$  that lay in the depth interval $(d^*,L(t^*)]$, is at least $d$, implying that $A \geq d$. Hence,
$$A \geq d \geq H.$$
Figure \ref{fig:lemma5_3} displays the interplay between these parameters in the context of an example attack.

\begin{figure}
    \centering
    \includegraphics[width=\linewidth]{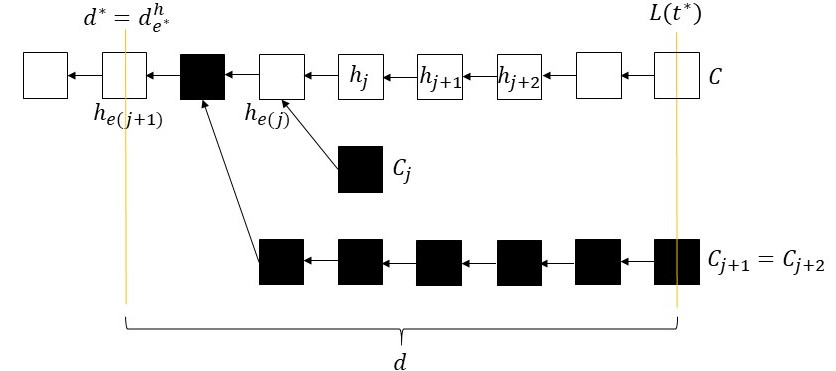}
    \caption{Chains $\C$, $\C_j$, $\C_{j+1}$, and, $\C_{j+2}$ for a sample attack $\pi$. In this example, $k=3$, $d=7$, $H=6$, and $A=8$. Chains $\C_{j+1}$ and $\C_{j+2}$ are the same, thus, $h_{e(j+1)}$ is the same honest block as $h_{e(j+2)}$, and $t_{j+1}=t_{j+2}$. Note that $t^*=t_{j+1}=t_{j+2}$ since $\C_{j+1}=\C_{j+2}$ is the last chain to catch up with $\C$. Similarly, $e^*=e(j+1)=e(j+2)$, and, $d^*=d^h_{e(j+1)}=d^h_{e(j+2)}$, as $h_{e(j+1)}=h_{e(j+2)}$ has depth smaller than $h_{e(j)}$.}
    \label{fig:lemma5_3}
\end{figure}

We now consider an adversary that follows strategy $\pi_{SZ}$. Again, let $\T^*$ be the blocktree built under $\pi_{SZ}$, and, define $d^{h*}_{e^*}$ as the depth of the block $h_{e^*}$ within $\T^*$. See Figure \ref{fig:lemma5_4} for a visual example of the $\pi_{SZ}$ attack. We next make the following observations using the properties of $\pi_{SZ}$: Via property (4) of $\pi_{SZ}$, $\C^*$ contains no adversary blocks at time $t^*$. Hence, at time $t^* \geq \tau^h_{j+k-1}$, $\C^*$ contains $h_m$, $m=j,..,j+k-1$ in a consecutive order, and, its depth is $L^*(t^*)=d^{h*}_{e^*}+H$. Via property (2), every adversary block mined after time $\tau^h_{e^*}$ is placed at a depth higher than $d^{h*}_{e^*}$. Via property (1), every adversary block mined after time $\tau^h_{e^*}$ is placed at a distinct depth. Hence, at time $t^*$, the deepest adversary block has depth at least $d^{h*}_{e^*}+A$. Via property (3), the path from this deepest adversary block to the genesis does not include $h_j$. Hence, it does not include any of the honest blocks $h_m$, $m=j,..,j+k-1$, that builds on $h_j$. Consequently, by time $t^*$, the adversary following $\pi_{SZ}$, has a private chain that does not include any of the blocks $h_m$, $m=j,..,j+k-1$, and, is at depth at least $d^{h*}_{e^*}+A$.

\begin{figure}
    \centering
    \includegraphics[width=\linewidth]{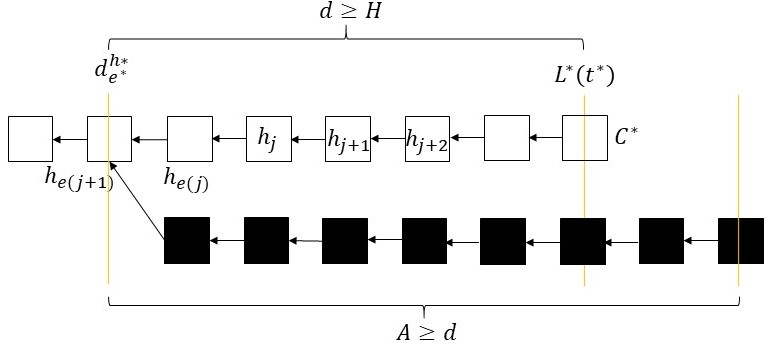}
    \caption{$\C^*$ and the private adversary chain under $\pi_{SZ}$ imposed on the same mining times as in Figure \ref{fig:lemma5_3}. Again, $k=3$, $d=7$, $H=6$, and, $A=8$. Adversary has a private chain at depth $A+d^{h*}_{e^*} = 8+d^{h*}_{e^*} > L^*(t^*) = H+d^{h*}_{e^*} = 6+d^{h*}_{e^*}$ at time $t^*$, and, the public longest chain $\C^*$ contains all of the attacked blocks $h_j$, $h_{j+1}$ and $h_{j+2}$ at time $t^*$.}
    \label{fig:lemma5_4}
\end{figure}
 
Finally, we have seen above that at time $t^*$, the public longest chain $C^*(t^*)$ contains all of the blocks $h_m$, $m=j,..,j+k-1$ and has depth $L^*(t^*)=d^{h*}_{e^*}+H$, whereas there exists a private adversary chain that does not include the blocks $h_m$, $m=j,..,j+k-1$, and, is at depth at least 
$$d^{h*}_{e^*}+A \geq d^{h*}_{e^*}+H = L^*(t^*).$$
Consequently, by broadcasting this private chain at time $t^*$, the adversary can prevent liveness for the $k$ consecutive honest blocks $h_j$ to $h_{j+k-1}$. This concludes the proof of part (ii) of the theorem.

\end{proof}

\subsection{Discussion on $\Delta > 0$}
\label{sec:lem:worst-attack-positive}

Theorem \ref{lem:worst-attack} shows that when $\Delta=0$, there exists an attack strategy, $\pi_{SZ}$, such that if any attack $\pi$ succeeds in preventing persistence for a block $b$ in the PoW model, this strategy also succeeds. Does such an attack strategy exist when $\Delta>0$ in the PoW model? Is private attack still the worst attack for every sequence of mining times when $\Delta>0$? Unfortunately, the answer is no: When $\Delta>0$, there does not exist a sample path worst attack. This is shown by the following lemma:

\begin{lemma}
\label{no-worst-attack}
Consider attacks for preventing the persistence, with some parameter $k$, of some block $h_j$, and, define the worst attack as the strategy $\pi^*$ satsifying the following condition: If some strategy $\pi \neq \pi^*$ succeeds under a sequence of mining times, then $\pi^*$ also succeeds under the same sequence except on a measure-zero set of sequences.
Then, when $\Delta>0$, and, $\lambda_a < \lambda_h/(1+\Delta \lambda_h)$, there does not exist a worst attack.
\end{lemma}

\begin{proof}
Proof is by contradiction. First, let $S_1$ be the set of mining time sequences for the blocks preceding $h_j$ such that $h_{j-1}$ is a loner, no adversary block is mined during the time interval $[\tau^h_{j-1},\tau^h_j]$, and, for any $i$, $0 \leq i<j-1$, $D_h(\tau^h_{j-1}-\Delta)-D_h(\tau^h_i+\Delta)$ is greater than the number of adversary arrivals during the time period $[\tau^h_i,\tau^h_{j-1}]$. ($D_h$ was defined previously in section 3.2.) Note that since this is a necessary condition for $h_{j-1}$ to be a Nakamoto block, and, $h_{j-1}$ is a Nakamoto block with positive probability when $\lambda_a < \lambda_h/(1+\Delta \lambda_h)$, there exists a constant $c>0$ such that $P(S_1)\geq c$ for all $j$. 

Second, consider the following set of mining times for the next three blocks that arrive after $h_j$:
\begin{itemize}
    \item Let $b$, $h_{j+1}$, and $b'$ denote these blocks in order of their mining times.
    \item $b$ is an adversary block and $h_{j+1}$ is an honest block.
    \item Mining time of $b$ satisfies the following equation: $$\tau^h_j < \tau_b < \tau^h_j+\Delta.$$ \item $b'$ is mined after time $\tau^h_{j+1}+\Delta$.
\end{itemize}
Now, depending on the mining time of $\tau^h_{j+1}$, we have two different sets of mining time sequences, $S_2$ and $S'_2$. The condition on $h_{j+1}$ which differentiates these two sets is given below:
\begin{itemize}
    \item $S_2$: $\tau_b < \tau^h_{j+1} < \tau^h_j + \Delta$
    \item $S'_2$: $\tau^h_j+\Delta < \tau^h_{j+1}$ 
\end{itemize}
We next consider the sets $S_1 \text{x} S_2$ and $S_1 \text{x} S'_2$. For the sake of simplicity, let's call any arbitrary sequence from $S_1 \text{x} S_2$, \textbf{sequence 1}, and, any arbitrary sequence from $S_1 \text{x} S'_2$, \textbf{sequence 2}. 

Now, for the sake of contradiction, assume that there exists a worst attack $\pi^*$ that aims to prevent the persistence of block $h_j$. Consider an arbitrary sequence of mining times from the set $S_1 \text{x} S_2 \cup S_1 \text{x} S'_2$. Via the definition of the set $S_1$, no matter what $\pi^*$ does, the deepest adversary block at time $\tau^h_{j-1}$ has depth smaller than $d^h_{j-1}$. Then, to prevent the persistence of $h_j$, $\pi^*$ builds two parallel chains starting at block $h_{j-1}$, only one of which contains $h_j$. Let $\C$ be the chain containing $h_j$ and let $\C'$ be the other parallel chain. It also delays the broadcast of block $h_j$ by $\Delta$ so that if $h_{j+1}$ is mined within $\Delta$ time of $h_j$, it is placed within the chain $\C'$, at the same depth as $h_j$. However, when block $b$ is mined, there are two distinct actions that $\pi^*$ might follow: 
\begin{enumerate}
    \item \textbf{Action 1:} Choose $h_j$ as $b$'s parent. Keep $b$ private until at least time $\tau^h_{j+1}+\Delta$.
    \item \textbf{Action 2:} Choose $h_{j-1}$ as $b$'s parent.
\end{enumerate}
(Note that a worst attack will not mine $b$ on a block preceding $h_{j-1}$.) See Figure \ref{fig:lemma5_5} for the effects of these actions on the blocktree under the sequences 1 and 2.

\begin{figure}
    \centering
    \includegraphics[width=\linewidth]{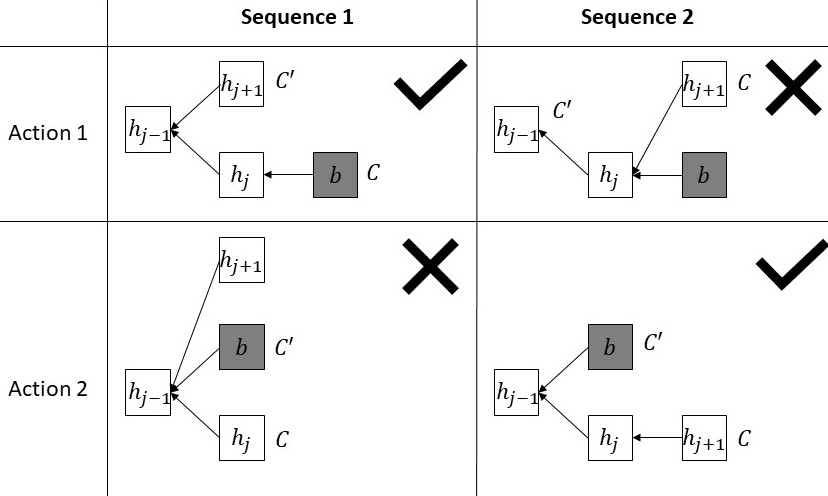}
    \caption{Blocktree for actions 1 and 2 under the sequences 1 and 2. Optimal actions for each sequence are marked with a tick.}
    \label{fig:lemma5_5}
\end{figure}

Now, assume that $\pi^*$ follows action 1. Then, under sequence 1, the optimal behavior for $\pi^*$ is to broadcast $h_{j+1}$ before $h_j$ becomes public at time $\tau^h_j+\Delta$, and, to prompt the honest miners to keep mining on $h_{j+1}$. Then, $\C'$ becomes the public longest chain, and, the adversary can balance the chains $\C$ and $\C'$ in the future using the private block $b$. However, if $\pi^*$ follows action 2, then, under sequence 1, $\C$ would not be leading $\C'$ via the private block, thus, making it harder for the adversary to maintain a balance between these two chains in the future. Hence, under sequence 1, for any sequence of mining times for the blocks after $h_{j+1}$, if $\pi^*$ following action 2 prevents the persistence of block $h_j$, so does $\pi^*$ following action 1. On the other hand, there exists a set $S_3$ of mining time sequences for the blocks after $h_{j+1}$ such that $P(S_3)>0$, and, under the sequences in $S_3$, following action 1 prevents the persistence of block $h_j$ whereas following action 2 does not. Since $P(S_1)\geq c>0$ for all $j$, the set $S_1 \text{x} S_2 \text{x} S_3$ has positive probability. Consequently, the worst attack $\pi^*$ does not follow action 2, implying that it follows action 1.

Next, observe that under sequence 2, $h_{j+1}$ comes to a higher depth than $h_j$. Hence, the optimal action for $\pi^*$ under sequence 2 is to follow action 2 as it enables the adversary to extend $\C'$ by one block using $b$. Action 1, on the other hand, does not help the adversary in its endeavor to maintain two parallel chains from block $h_{j-1}$ as demonstrated by Figure \ref{fig:lemma5_5}. Then, under sequence 2, for any sequence of mining times for the blocks after $h_{j+1}$, if $\pi^*$ following action 1 prevents the persistence of block $h_j$, so does $\pi^*$ following action 2. On the other hand, there exists a set $S'_3$ of mining time sequences for the blocks after $h_{j+1}$ such that $P(S'_3)>0$, and, under the sequences in $S'_3$, following action 2 prevents the persistence of block $h_j$ whereas following action 1 does not. Since $P(S_1)\geq c>0$ for all $j$, the set $S_1 \text{x} S'_2 \text{x} S'_3$ has positive probability. Consequently, the worst attack $\pi^*$ does not follow action 1, implying that it follows action 2. However, this is a contradiction as the worst attack $\pi^*$ can choose only one of the actions 1 and 2. Hence, there does not exist a worst attack $\pi^*$. 

\end{proof}

Finally, via the lemma \ref{no-worst-attack}, we observe that, for any given attack strategy $\pi$, there exists a set of mining time sequences with positive probability (which can be very small) under which $\pi$ is dominated by some other attack strategy. However, it is important to note that if we fix $\Delta$ to be some finite value and $\pi$ to be the private attack, probability of such atypical sets of mining time sequences go to zero as the parameter for persistence, $k$, goes to infinity. This is because, as we have seen in the previous sections, the private attack is the worst attack in terms of achieving the security threshold.